\documentclass[runningheads]{llncs}
\usepackage{graphicx}

\usepackage[firstpage]{draftwatermark}

\SetWatermarkText{\raisebox{12.5cm}{%
  \includegraphics{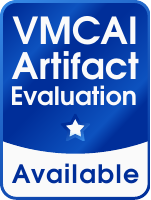} \hspace{6.3cm} \includegraphics{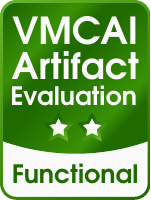}%
}}
\SetWatermarkAngle{0}

\usepackage[T1]{fontenc}

\usepackage{amsmath, amssymb, amsfonts}
\usepackage{mathtools}

\let\vec\relax
\usepackage{MnSymbol}
\usepackage{booktabs}

\usepackage{thm-restate}

\usepackage{dcolumn}
\newcolumntype{d}{D{/}{~}{5.4} }
\newcolumntype{f}{D{/}{~}{6.3} }

\usepackage{arydshln}

\usepackage{enumitem}

\usepackage[usenames,dvipsnames]{xcolor}

\usepackage{subcaption}
\usepackage[noend,linesnumbered,ruled]{algorithm2e}
\SetKw{Continue}{continue}
\SetKw{Break}{break}
\SetKwComment{Comment}{$\triangleright$\ }{}

\usepackage{ragged2e}

\usepackage{tikz}
\usetikzlibrary{matrix}
\usetikzlibrary{fit}
\usetikzlibrary{automata}
\usetikzlibrary{shapes}
\usetikzlibrary{arrows}
\usetikzlibrary{positioning}
\usetikzlibrary{calc}
\usetikzlibrary{trees}

\tikzset{>=latex}

\tikzset{
  event_structure/.style={
    ->,
    nodes={ minimum size = .5cm}
  },
}

\tikzstyle{level 1}=[level distance=3.5cm, sibling distance=3.5cm]
\tikzstyle{level 2}=[level distance=3.5cm, sibling distance=2cm]
\tikzstyle{small}=[scale=0.7]
\tikzstyle{event}=[draw,rectangle,prefix after command= {\pgfextra{\tikzset{every label/.style={small}}}}]
\tikzstyle{conflict}=[dashed]
\tikzstyle{invisible}=[draw opacity=0]
\tikzstyle{conf}=[line width=0.5mm]
\tikzstyle{edge from parent}=[draw,->]

\usepackage[capitalise, noabbrev]{cleveref}

\crefname{appsec}{Appendix}{Appendices}

\usepackage{stackengine}
\stackMath

\newcommand\review[1]{#1}

\title{Language Inclusion for \\Finite Prime Event Structures
\thanks{
This work has received funding from the Electronic Component
Systems for European Leadership Joint Undertaking under grant
agreement No 737459 (project Productive4.0), which receives support from the European Union Horizon 2020
research and innovation program and Germany, Austria, France,
Czech Republic, Netherlands, Belgium, Spain, Greece, Sweden,
Italy, Ireland, Poland, Hungary, Portugal, Denmark, Finland, Luxembourg, Norway, Turkey, 
from the by ECSEL Joint Undertaking
under the project H2020 737469 AutoDrive --- Advancing failaware,
fail-safe, and fail-operational electronic components, systems, and
architectures for fully automated driving to make future mobility
safer, affordable, and end-user acceptable,
from the LogiCS doctoral program W1255-N23 of the Austrian Science Fund (FWF) and by the Vienna Science and Technology Fund
(WWTF) through the projects Heisenbugs project VRG11-005.
}
}

\author{Andreas Fellner\inst{1,2}\orcidID{0000-0002-3618-2251} \and
Thorsten Tarrach\inst{1}\orcidID{0000-0003-4409-8487} \and
Georg Weissenbacher\inst{2}\orcidID{0000-0002-0143-632X}}

\author{Andreas Fellner\inst{1,2} \and
Thorsten Tarrach\inst{1} \and
Georg Weissenbacher\inst{2}}

\institute{AIT Austrian Institute of Technology\\1210 Vienna, Austria \and
TU Wien, 1040 Vienna, Austria \\
\email{andreas.fellner@ait.ac.at}}

\authorrunning{A. Fellner et al.}

\newcommand{\lang}{\mathcal{L}}
\newcommand{\aut}{\mathcal{A}}

\newcommand\jsucc[1]{\mbox{su}(#1)}

\newcommand\source[1]{s(#1)}
\newcommand\target[1]{t(#1)}

\newcommand{\hist}[1]{\lceil #1 \rceil}

\newcommand\trace\tau

\newcommand{\les}{\mathcal{E}}

\newcommand{\kw}[1]{\mathit{#1}}
\newcommand{\var}[1]{\mathit{#1}}

\newcommand{\dsucc}{\kw{dsucc}}

\newcommand{\true}{\mathtt{true}}
\newcommand{\false}{\mathtt{false}}

\newcommand{\none}{\mathtt{None}}

\newcommand{\frontier}{\var{frontier}}
\newcommand{\push}{\var{push}}
\newcommand{\pop}{\var{pop}}
\newcommand{\range}{\var{range}}

\newcommand{\leslabel}{\mathcal{X}}

\newcommand{\bhalf}{\mathit{bh}}

\makeatletter
\DeclareRobustCommand{\cev}[1]{%
  \mathpalette\do@cev{#1}%
}
\newcommand{\do@cev}[2]{%
  \fix@cev{#1}{+}%
  \reflectbox{$\m@th#1\vec{\reflectbox{$\fix@cev{#1}{-}\m@th#1#2\fix@cev{#1}{+}$}}$}%
  \fix@cev{#1}{-}%
}
\newcommand{\fix@cev}[2]{%
  \ifx#1\displaystyle
    \mkern#23mu
  \else
    \ifx#1\textstyle
      \mkern#23mu
    \else
      \ifx#1\scriptstyle
        \mkern#22mu
      \else
        \mkern#22mu
      \fi
    \fi
  \fi
}

\makeatother

\begin{document}

\date{\today}

\maketitle

\begin{abstract}

We study the problem of language inclusion between finite, labeled prime event structures.
Prime event structures are a formalism to compactly represent concurrent behavior of discrete systems.
A labeled prime event structure induces a language of sequences of labels produced by the represented system.
We study the problem of deciding inclusion and membership for languages encoded by finite prime event structures and provide complexity results for both problems.
We provide a family of examples where prime event structures are exponentially more succinct than formalisms that do not take concurrency into account. 
We provide a decision algorithm for language inclusion that exploits this succinctness.
%
%
%
Furthermore, we provide an implementation of the algorithm and an evaluation on a series of benchmarks.
Finally, we demonstrate how our results can be applied to mutation-based test case generation.
\end{abstract}


\section{Introduction}

Language inclusion is a fundamental problem in computer science which
arises in numerous application domains. In its most familiar form
the problem is instantiated with regular languages and finite
automata \cite{automata_book}; an incarnation frequently occurring
in formal verification and model checking is language inclusion
(and intersection, respectively) for $\omega$-regular languages
and B\"uchi automata \cite[Chapter 7]{modelchecking_book}.
In the latter application, the goal is to check whether a
transition system conforms to a specification given in
linear temporal logic. 
One challenge arising in automata-based model checking is that the
verification of concurrent systems relies on the explicit construction 
of a product automaton whose size can be exponential
in the number of processes. Partial Order Reduction 
(POR, see \cite{Valmari1989,Godefroid1996} and
\cite[Chapter 12]{modelchecking_book}, for instance)
addresses this problem by exploiting independence between
transitions to avoid the construction of the full
product automaton: the reduction identifies
equivalence classes of words in the language (i.e., executions)
obtained by reordering commutative edges/transitions \cite{Mazurkiewicz1986}
and restricts the exploration to representative members of these classes. 
POR in its simplest form can be used to check reachability and deadlock problems; 
for checking temporal logic properties only transitions
whose labels are ``invisible'' to the property are assumed to be
independent \cite[Chapter 12]{modelchecking_book}.
This renders the approach impractical for language
inclusion if the alphabets of both languages are the same, e.g., when
checking whether a modification is language-preserving --
a question arising in the applications that motivated our work (see below).

In this paper, we focus on language inclusion for finite, labeled prime
event structures, a representation of bounded executions of
concurrent systems in which dependence (and independence) of
transitions is made explicit.
This representation can be exponentially more
succinct than finite automata, as shown in \cref{sec:algorithm}:
there are event structures with $n$ events, 
such that the smallest NFA expressing the same language has at
least $2^n$ states.

We provide an analysis of the computational complexity of checking
language membership as well as inclusion between two event structures,
showing that the former is NP-complete and the latter is $\Pi^p_2$-complete
(\cref{sec:complexity}).
\review{
While a similar results to the former was proven earlier for
trace languages \cite{bertoni1989membership}, to the best of our knowledge, 
the latter result is novel even in the related domains of bounded trace languages and bounded labeled Petri nets.
}

\review{Besides showing the complexity of the decision problems, 
we provide a practical decision algorithm for solving event structure language inclusion in \cref{sec:algorithm}.}
By finding suitable embeddings of one event structure in another, the algorithm
determines whether the language of the former is included in the language
of the latter. The algorithm iteratively refines the event structure
whenever two labels occur unordered in the former structure 
but ordered in the latter. Moreover, the algorithm can provide
counterexamples to inclusion encoded as event
structures representing words that occur in the former language but
not in the latter.

\Cref{sec:experiments} provides a qualitative analysis of our
representation and an experimental evaluation that highlights
advantages and disadvantages of event structures
in comparison to an automaton-based representation (for which
language inclusion is PSPACE-complete). 

\review{Our inclusion algorithm decides whether two systems, represented as
event structures, have the same behavior in terms of bounded
words over a common vocabulary.} This scenario arises in a range of
applications: refinement or model checking, where an implementation is
compared against a specification; upgrade or regression checking,
where a fixed version of a software is compared against the original
version; or mutation-based test case generation, where a small
modification (or bug) is introduced in code to obtain a ``mutant''
of the original program, and the counterexample to inclusion then
represents a test case which discriminates between mutant and
original. 
We use the latter scenario, which motivated our research on
language inclusion, as an exemplary application
of our approach in our experiments (\cref{sec:experiments}).



\section{Preliminaries}

In this section we introduce labeled prime event structures.
%
%
Throughout this work, we assume that every set of labels $\leslabel$ contains a distinct label $\varepsilon$, 
which denotes the empty symbol.
Concatenation of $\varepsilon$ to a word does not change the word.

\begin{definition}[FLES]
Given a set of labels $\leslabel$, a finite, $\leslabel$-labeled prime event structure (FLES) is a tuple $\les := \langle E, <, \#, h\rangle$
where $E$ is a finite set of events, $<\, \subseteq E \times E$ is a strict partial order on $E$, called \emph{causality relation},
$h: E \rightarrow \leslabel$ labels every event with an element of $\leslabel$, 
and $\# \subseteq E\times E$ is the symmetric, irreflexive \emph{conflict relation} \review{that is \emph{closed under} $<$, i.e.}
for all $e,e',e'' \in E$, if $e \# e'$ and $e' < e''$, then $e \# e''$.
\end{definition}
For an event $e$, we use $\hist e$ to denote the history of $e$ as the set of events that must happen before $e$ according to $<$, formally $\hist e := \{e' \in E \mid e' < e\}$. 
We require that there is a special event $\bot \in E$, such that $\hist\bot = \emptyset$, for all events $e \in E:\ \bot < e$, and $h(\bot) = \varepsilon$.
We define the direct successors $\dsucc$ of event $e$ as the set of events that depend on $e$ without there being another event in-between, formally $\dsucc(e)=\{e'\in E \mid e < e' \wedge \nexists e'':\ e < e'' < e'\}$.
We say that two events $e, e' \in E$ are \emph{concurrent} if $e \neq e'$, not $(e < e')$, not $(e > e')$, and not $(e \# e')$.

A central concept in assigning event structures a semantic is the notion of configurations:

\begin{definition}[Configuration]
\review{
For a FLES $\les := \langle E, <, \#, h\rangle$, a
\emph{configuration} of $\les$ is a set of events
$C~=~\{e_1,\ldots,e_n\}~\subseteq~E$ that is both
}
\begin{itemize}
	\item Left closed: $\forall e \in C: \forall e' \in E $ such
          that $e' < e \implies e' \in C$, and
	\item Conflict free: $\forall e,e' \in C: \neg (e \# e')$
\end{itemize}

\end{definition}

A configuration $C$ is \emph{maximal}, if there is no configuration $C'$ such that \review{$C \subseteq C'$ and $C \neq C'$}.
We denote by $\mathcal{MC}(\les)$ the set of all maximal configurations of an event structure $\les$.
A \emph{trace} $\trace$ of $C$ is a sequence of events $\langle e_1,\ldots,e_n \rangle$, where every event $e\in C$ occurs exactly once in the sequence and for all $e_i,e_j\in \trace:\  e_i<e_j\implies i < j$.
We denote the set of all traces of a configuration $C$ with $T(C)$.
Let $f: C \rightarrow X$ be a mapping on $C$ to some set $X$.
For a trace $\trace$ of $C$, we denote by $f(\trace)$ the sequence resulting from point-wise application of $f$ on the elements of $\trace$.
Finally, we extend $T$ to event structures by defining it as the union
of traces over all maximal configurations.
That is, $T(\les) := \bigcup_{C \in \mathcal{MC}(\les)} T(C)$.

A finite, labeled prime event structure $\les$ represents a finite set of bounded words over an alphabet $\leslabel$, where 
the bound for the length of words is given by the size of the largest maximal configuration.
We call this set the language $\lang(\les)$. 

\begin{definition}[Language of $C$ and $\les$]
The language of configuration $C$ of $\les$ is $\lang(C) := \{ h(\trace) \mid \trace \in T(C)\}$. 
The language of $\les$ is $\lang(\les) := \{h(\trace) \mid \trace \in T(\les)\}$.
\end{definition}

To illustrate this definition we give a small example.

\begin{figure}[tbp]
\begin{subfigure}[t]{0.49\textwidth}
\centering
\begin{tikzpicture}[grow'=right, sloped,scale=0.5,level 1/.style={level distance=3cm, sibling distance=3cm}]
	\node[event] {$\bot$}
	child {
		node[event,label={above:A}] (e1) {$e_1$}
		edge from parent
	}
	child {
		node[event,label={below:B}] (e3) {$e_2$}
		edge from parent
	};
	
\end{tikzpicture}
\caption{LES with one maximal configuration}
\label{fig:ex-noconf}
\end{subfigure}\hfill
\begin{subfigure}[t]{0.49\textwidth}
\centering
\begin{tikzpicture}[grow'=right, sloped,scale=0.5,level 1/.style={level distance=3cm, sibling distance=3cm},level 2/.style={level distance=3cm, sibling distance=3cm}]
	\node[event] {$\bot$}
	child {
		node[event,label={above:A}] (e1) {$e_1$}
		edge from parent
		child {
			node[event,label={above:B}] {$e_2$}
			edge from parent
		}
	}
	child {
		node[event,label={below:B}] (e2) {$e_3$}
		edge from parent
		child {
			node[event,label={below:A}] {$e_4$}
			edge from parent
		}
	};
	
	\draw[conflict] (e1) -- (e3);
\end{tikzpicture}
\caption{Event structure with conflicts}
\label{fig:ex-conf}
\end{subfigure}\hfill
\caption{Event structures}
\label{fig:ex-simple}
\end{figure}
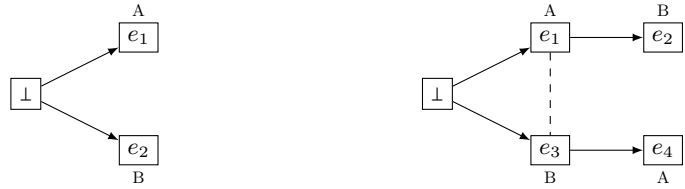

\begin{example}[Event structure and configurations]
We show two event structures in \cref{fig:ex-noconf,fig:ex-conf}. 
Boxes depict events.
Inside every box is its event's identifier, above or below the box is its event's label. If there is no label we implicitly assume the label to be $\varepsilon$.
Solid arrows depict direct successors of an event.
Dashed lines depict immediate conflicts.
Two events $e,e'$ are in immediate conflict if $e\# e'$ and there are no $e_1,e_2 \in E$ such that $e_1 < e \wedge e_1 \# e'$ or $e_2 < e' \wedge e\#e_2$.
For better readability, we omit all other causalities and conflicts.

\Cref{fig:ex-noconf,fig:ex-conf} both represent the language $\{AB,BA\}$. The event structure in \cref{fig:ex-noconf} has a single maximal configuration consisting of events $\{\bot,e_1,e_2\}$. The event structure in \cref{fig:ex-conf} has two maximal configurations: $\{\bot,e_1,e_2\}$ and $\{\bot,e_3,e_4\}$ (due to the conflict between $e_1$ and $e_3$ these two events cannot appear in the same configuration).
\end{example}



\section{Language Inclusion Problem and Complexity Results}
\label{sec:complexity}

\begin{figure}[tbp]
\centering
\begin{tikzpicture}[grow'=right, sloped,scale=0.5,level 1/.style={level distance=3cm, sibling distance=3cm}]
	\node[event] (bot) {$\bot$}
	child {
		node[event] (ef11) {$e_{f_1,1}$}
		edge from parent
    child {
      node[event] (ef21) {$e_{f_2,1}$}
      edge from parent[draw=none]
      child {
        node[event] (efm1) {$e_{f_k,1}$}
        edge from parent[draw=none]
      }
    }
	}
  child {
		node[event] (ef12) {$e_{f_1,2}$}
		edge from parent
    child {
      node[event] (ef22) {$e_{f_2,2}$}
      edge from parent[draw=none]
      child {
        node[event] (efm2) {$e_{f_k,2}$}
        edge from parent[draw=none]
        child {
          node[event,xshift=.5cm,label={above:$x$}] (ef21m2) {$e_{f_2,f_k,1}$}
          edge from parent
        }
      }
    }
	}
	child {
		node[event] (ef1n) {$e_{f_1,n}$}
		edge from parent
    child {
      node[event] (ef2n) {$e_{f_2,n}$}
      edge from parent[draw=none]
      child {
        node[event] (efmn) {$e_{f_k,n}$}
        edge from parent[draw=none]
        child {
          node[event,xshift=.5cm,label={above:$x$}] (ef2xmn1) {$e_{f_2,f_k,n-1}$}
          edge from parent
        }
      }
    }
	};
  \node[event,below=.9cm of ef2xmn1,label={above:$x$}] (ef2xmn) {$e_{f_2,f_k,n}$};
  \node[below=.3cm of ef22,invisible] (efn12) {};
  \draw[conflict] (ef11)--(ef21);
  \draw[conflict] (ef12)--(ef22);
  \draw[conflict] (ef1n)--(ef2n);
  \draw[conflict] (ef11)--(ef12);
  \draw[conflict] (ef21)--(ef22);
  \draw[conflict] (efm1)--(efm2);
  
  \draw[conflict,bend left] (ef11) to (efm1);
  \draw[conflict,bend left] (ef12) to (efm2);
  \draw[conflict,bend left] (ef1n) to (efmn);
  
  \draw[conflict,bend right] (ef11) to (ef1n);
  \draw[conflict,bend right] (ef21) to (ef2n);
  \draw[conflict,bend right] (efm1) to (efmn);
  
  \draw[invisible] (ef21) --node[auto=false]{\ldots} (efm1);
  \draw[invisible] (ef22) --node[auto=false]{\ldots} (efm2);
  \draw[invisible] (ef2n) --node[auto=false]{\ldots} (efmn);
  \draw[invisible] (ef12) --node[auto=false]{\ldots} (ef1n);
  \draw[invisible] (ef22) --node[auto=false]{\ldots} (ef2n);
  \draw[invisible] (efm2) --node[auto=false]{\ldots} (efmn);
  \draw[invisible] (ef21m2) --node[auto=false]{\ldots} (ef2xmn1);
  \draw[edge from parent] (ef21) -- (ef21m2);
  \draw[edge from parent] (efn12) -- (ef2xmn1);
  \draw[edge from parent] (ef2n) -- (ef2xmn);
  \draw[edge from parent] (efm1) -- ++(16mm,0) |- (ef2xmn);
  
  \begingroup
  \node[above=0.6cm of ef21m2,yshift=1mm,align=center] (node2) {\scriptsize if $\target{f_2} = \source{f_k}$:};
  \node[above=1.3mm of node2,xshift=1mm] (node3) {\scriptsize \bf Connected edges};
  
  \draw ($ (node2.north east) + (5mm,3mm) $) rectangle ($ (ef2xmn.south west) - (10mm,3mm) $);
  \endgroup
  
  \begingroup
  \node[right=9cm of bot,yshift=1mm,align=left] (node1) {\scriptsize if $\target{f_p} = \target{f_q}$:\\ $\forall i,j$};
  \node[above=1.8mm of node1,xshift=5mm] (node0) {\scriptsize \bf Conditional Conflicts};
  
  \node[event,below=0mm of node1,xshift=-5mm] (efp) {$e_{f_p,i}$};
  \node[event,below=0mm of node1,xshift=10mm] (efq) {$e_{f_q,j}$};
  \draw[conflict] (efp) -- (efq);
  \draw ($ (node1.north east) + (10mm,3mm) $) rectangle ($ (efp.south west) - (3mm,3mm) $);
  \endgroup

\node[below=3cm of bot,xshift=4cm,align=left] {All $e_{f_1,1},\ldots,e_{f_k,1},e_{f_1,n},\ldots,e_{f_k,n}$ are causally related to $\bot$};
	
\end{tikzpicture}
\caption{$\les^G$ for \cref{thm:membership_np_hard}}
\label{fig:ex-three}
\end{figure}

The language inclusion problem for two event structures $\les_1,  \les_2$ is to decide whether $\lang(\les_1) \subseteq \lang(\les_2)$.
In this section we prove a complexity bound for the language inclusion problem. 
As an intermediate step we look at the membership problem.

\subsection{Language Membership is NP-complete}

The \emph{finite prime event structure language membership problem} for word $w$ and FLES $\les$ is the problem of deciding whether $w \in \lang(\les)$.
Surprisingly, deciding membership is NP-complete. 
In contrast, trace membership $\trace \in T(\les)$ can be decided in polynomial time.
Trace membership can be decided simply by verifying that the set of events of $\trace$ forms a maximal configuration of $\les$, 
which requires to verify left-closure, conflict-freedom, and maximality.
All of those can be checked in polynomial time (\review{linear time, assuming linear conflict lookup}).

Intuitively, the hardness of language membership comes from the fact that the labeling function does not need to be injective and the role of conflicts, 
which together rule out a greedy algorithm that consumes the word in question symbol by symbol in a unique way.

\begin{theorem}
\label{thm:membership_in_np}

Finite prime event structure language membership is in NP.

\end{theorem}

\begin{proof}

Let $\les = \langle E, <, \#,h \rangle$ be an $\leslabel$-labeled FLES and $w = \langle \sigma_1,\ldots,\sigma_n \rangle \in \leslabel^*$ be a word.
A trace $\trace$ is a polynomially sized certificate for $w \in \lang(\les)$.
Checking that $\trace\in T(\les)$ can be done in polynomial time, and 
checking whether $h(\trace) = w$ can be done in linear time.
\qed

\end{proof}

To prove NP-hardness we reduce the Hamiltonian cycle (HC) problem  to the membership problem. 
HC is known to be NP-hard~\cite{karp1972reducibility}. It is the problem of deciding whether for a directed graph there exists a path that visits all vertices once and that ends in the vertex it started.
We use $\source{f}$ and $\target{f}$ to denote the source and target of a directed edge $f$.

\begin{restatable}{theorem}{membershipnphard}
\label{lem:cycle_red}
\label{thm:membership_np_hard}
Finite prime event structure language membership is NP-hard.
\end{restatable}

\begin{proof}
\review{
For a directed graph $G=(\{v_1,\ldots,v_n\},\{f_1,\ldots,f_k\})$ we construct an event structure $\les^G$, such that $x^n \in \lang(\les^G)$ iff $G$ has a Hamiltonian cycle. $\les^G$ is shown in \cref{fig:ex-three} and we present the main arguments why this reduction is correct here. 
A detailed, formal proof is given in \cref{sec:langincproof}.

Configurations of the event structure encode a sequence of $n$ edges.
If event $e_{f,i}$ is included in the configuration it means that edge $f$ is at position $i$ in the sequence of edges. 
To ensure that every vertex is visited, edges with the same target are in conflict. 
Since $n$ edges need to be selected, there are $n$ vertices, and every vertex is a target of some selected edge, every vertex is visited once by the selected edges.
To ensure that the sequence of edges actually forms a cycle they need to be connected. 
Events $e_{f,i}$ and $e_{f',{i+1} \mathsf{mod} n}$ for which the target of $f$ is the source of $f'$ cause an $x$-labeled event $e_{f,f',i}$.
Therefore, only configurations that represent a cycle form the word $x^n$.

In summary, checking the membership of $x^n$ 
amounts to checking whether there exists a Hamiltonian cycle in $G$.
The reduction clearly is polynomial.
}
\qed
\end{proof}

\subsection{Language Inclusion is $\Pi^p_2$-complete}

The \emph{finite prime event structure language inclusion problem} for FLES $\les_1$ and $\les_2$ is the problem of deciding whether $\lang(\les_1) \subseteq \lang(\les_2)$.

$\Pi^p_2$ is a complexity class from the polynomial hierarchy. It intuitively represents a $\forall\exists$ quantifier alternation.
To show inclusion, we use the definition of $\Pi^p_2$ given by Wrathall~\cite{wrathall1976complete},
providing semantics for the complexity class in terms of formal languages.
These languages should not be confused with the particular type of languages we discuss in this work.
In contrast, such languages encode problem instances and candidate witnesses.

Formally, a language $L$ is in $\Pi^p_2$ iff there exists a polynomially decidable language $L'$, such that $x\in L \Leftrightarrow \forall y_1 \exists y_2 [\langle x,y_1,y_2\rangle \in L']$.
A language $L'$ is polynomially decidable if $w \in L'$ can be decided in polynomial time.
The $x$ represents an encoding of the problem instance as a string.
The $y_1$ and $y_2$ represent string encodings of witnesses to a sub-problem.


We fix two $\leslabel$-labeled FLES $\les_1 = \langle E_1, <_1, \#_1, h_1\rangle$ and $\les_2 = \langle E_2, <_2, \#_2, h_2\rangle$.

\begin{theorem}
\label{thm:inclusion_in_pip2}

Finite prime event structure language inclusion is in $\Pi^p_2$.

\end{theorem}

\begin{proof}

Language inclusion $\lang(\les_1) \subseteq \lang(\les_2)$ amounts to checking whether $\forall w\in \lang(\les_1)\Rightarrow w \in \lang(\les_2)$.
In terms of traces this can be expressed as
$\forall \trace_1 \in T(\les_1) .\ \exists \trace_2 \in T(\les_2) .\ h_1(\trace_1)=h_2(\trace_2)$, meaning that for every trace in $\les_1$ there has to be a trace in $\les_2$ corresponding to the same word in the common alphabet $\leslabel$.

We define $L := \{\langle \les_1, \les_2 \rangle \mid \lang(\les_1) \subseteq \lang(\les_2)\}$ and $L' := \{\langle \langle \les_1, \les_2 \rangle,\trace_1,\trace_2\rangle \mid \trace_1 \in T(\les_1) \Rightarrow \big( h_1(\trace_1)=h_2(\trace_2) \wedge \trace_2 \in T(\les_2) \big)\}$.
By the argument above, we obtain the desired form $x\in L$ iff $\forall y_1 \exists y_2 [\langle x,y_1,y_2\rangle \in L']$ to show $\Pi^p_2$ inclusion.
Furthermore, $L'$ can be decided deterministically in polynomial time, because trace membership, as well as label equality, can be decided in polynomial time.
\qed
\end{proof}

To show $\Pi^p_2$ hardness, we present a reduction from the Dynamic Hamiltonian Cycle (DHC) problem to the finite prime event structure language inclusion problem.
Given an undirected graph $G=(V,F)$ and a set $B \subseteq F$, graph $G$ and $B$ form a DHC if for every set $D \subseteq B$ with $|D| \leq |B|/2$, the graph $G_D=(V,F\setminus D)$ has a Hamiltonian cycle. We define $n:=|V|$, $k:=|F|$, $m:=|B|$, and $bh:=\lfloor |B|/2 \rfloor$.
Essentially DHC, in comparison to HC, has an additional universal quantifier over subsets of $B$.
DHC is known to be $\Pi^p_2$-complete \cite{ko1995complexity}.

\begin{restatable}{theorem}{lemdyncycle}
\label{lemma:dyn_ham_cycle}
\label{thm:inclusion_pip2-hard}
Finite prime event structure language inclusion is $\Pi^p_2$-hard.
\end{restatable}

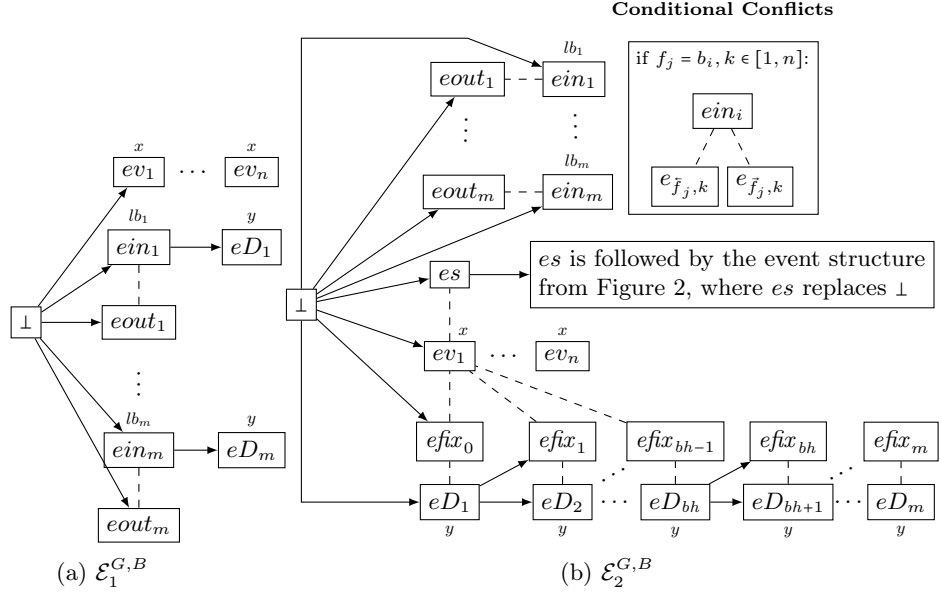
\begin{figure}[tb]
\begin{subfigure}[t]{0.2\textwidth}
\centering
\begin{tikzpicture}[grow'=right, sloped,scale=0.5,level 1/.style={sibling distance=2cm,level distance=3cm},level 2/.style={sibling distance=2cm,level distance=3cm}]
	\node[event] {$\bot$}
	child {
		node[event,label={above:$x$}] (ev1) {$ev_1$}
		edge from parent
    child {
      node[event,label={above:$x$}] (evn) {$ev_n$}
      edge from parent[draw=none]
    }
	}
  child {
		node[event,label={above:$lb_1$}] (eb1) {$ein_1$}
		edge from parent
    child {
      node[event,label={above:$y$}] (ebb1) {$eD_1$}
      edge from parent
    }
	}
  child {
		node[event] (ebbb1) {$eout_1$}
		edge from parent
	}
  child {
		node[event,label={above:$lb_m$},yshift=-.7cm] (ebm) {$ein_m$}
		edge from parent
    child {
      node[event,label={above:$y$}] (ebbm) {$eD_m$}
      edge from parent
    }
	}
  child {
		node[event,yshift=-.7cm] (ebbbm) {$eout_m$}
		edge from parent
	};
	
  \draw[invisible] (ev1) --node[auto=false]{\ldots} (evn);
  \draw[conflict] (eb1) -- (ebbb1);
  \draw[conflict] (ebm) -- (ebbbm);
  \draw[invisible] (ebbb1) --node[auto=false]{\ldots} (ebm);
\end{tikzpicture}
\caption{$\les^{G,B}_1$}
\label{fig:ex-g1}
\end{subfigure}\hfill
\begin{subfigure}[t]{0.7\textwidth}
\centering
\begin{tikzpicture}[grow'=right, sloped,scale=0.5,level 1/.style={sibling distance=2.5cm,level distance=3cm},level 2/.style={sibling distance=2.5cm,level distance=3cm}]
	\node[event] (bot) {$\bot$};
  \node[event,label={above:$lb_1$},right=3cm of bot,yshift=3.0cm] (eb1) {$ein_1$};
  \node[event,left=0.5cm of eb1] (ebbb1) {$eout_1$};
		
	\node[event,label={above:$lb_m$},right=3cm of bot,yshift=1.5cm] (ebm) {$ein_m$};
  \node[event,left=0.5cm of ebm] (ebbbm) {$eout_m$};
	\node[above=.2cm of ebbb1] (silent) {};
  \draw[edge from parent] (bot) |- (silent.center) -- (eb1);
  \draw[edge from parent] (bot) -- (ebbb1);
  \draw[edge from parent] (bot) -- (ebm);
  \draw[edge from parent] (bot) -- (ebbbm);
  \draw[invisible] (eb1) --node[auto=false,xshift=-1mm]{\ldots} (ebm);
  \draw[invisible] (ebbb1) --node[auto=false,xshift=1mm]{\ldots} (ebbbm);
  
  \node[event,right=1.5cm of bot,yshift=0.4cm] (es) {$es$};
  
  \draw[edge from parent] (bot) -- (es);
  
  \node[event,label={80:$x$},below=0.7cm of es] (ev1) {$ev_1$}
  child {
    node[event,label={above:$x$}] (evn) {$ev_n$}
    edge from parent[draw=none]
  };
  \draw[conflict] (ev1) -- (es);
	
	\draw[edge from parent] (bot) -- (ev1);
	
  \node[event,label={below:$y$},below=1.5cm of ev1] (ey1) {$eD_1$}
  child {
    node[event,label={below:$y$}] (ey2) {$eD_2$}
    edge from parent
    child {
      node[event,label={below:$y$}] (eyhalf) {$eD_{\bhalf}$}
      edge from parent[draw=none]
      child {
        node[event,label={below:$y$}] (eyhalf1) {$eD_{\bhalf+1}$}
        edge from parent
        child {
          node[event,label={below:$y$}] (eym) {$eD_m$}
          edge from parent[draw=none]
        }
      }
    }
  };
  \draw[edge from parent] (bot) |- (ey1);

  \node[event,above=0.3cm of ey1] (eyy1) {$\mathit{efix}_0$};
  \draw[edge from parent] (bot) -- (eyy1);
  \draw[conflict] (ey1) -- (eyy1);
  \draw[conflict] (eyy1) -- (ev1);
  \node[event,above=0.3cm of ey2] (eyy2) {$\mathit{efix}_1$};
  \draw[edge from parent] (ey1) -- (eyy2);
  \draw[conflict] (ey2) -- (eyy2);
  \draw[conflict] (eyy2) -- (ev1);
  \node[event,above=0.3cm of eyhalf] (eyyhalf) {$\mathit{efix}_{\bhalf-1}$};
  \draw[conflict] (eyhalf) -- (eyyhalf);
  \draw[invisible] (ey2) --node[auto=false]{\ldots} (eyyhalf);
  \draw[conflict] (eyyhalf) -- (ev1);
  
  \node[event,above=0.3cm of eyhalf1] (eyyhalf1) {$\mathit{efix}_{\bhalf}$};
  \draw[edge from parent] (eyhalf) -- (eyyhalf1);
  \draw[conflict] (eyhalf1) -- (eyyhalf1);
  \node[event,above=0.3cm of eym] (eyym) {$\mathit{efix}_m$};
  \draw[conflict] (eym) -- (eyym);
  
  \draw[invisible] (ey2) --node[auto=false]{\ldots} (eyhalf);
  \draw[invisible] (eyhalf1) --node[auto=false]{\ldots} (eym);
  \draw[invisible] (eyhalf1) --node[auto=false]{\ldots} (eyym);

  \draw[conflict] (eb1) -- (ebbb1);
  
  \draw[conflict] (ebm) -- (ebbbm);

  \draw[invisible] (ev1) --node[auto=false]{\ldots} (evn);
  
  \begingroup
  \node[above=2.8cm of bot,xshift=5.6cm] (node1) {\scriptsize if $f_j = b_i, k \in [1,n]$:};
  \node[above=2.4mm of node1] (node0) {\scriptsize \bf Conditional Conflicts};
  \node[event,below=2mm of node1] (seini) {$ein_i$};
  \node[event,below=5mm of seini,xshift=5mm] (sevfj) {$e_{\vec{f}_j,k}$};
  \node[event,below=5mm of seini,xshift=-5mm] (secfj) {$e_{\cev{f}_j,k}$};
  \draw[conflict] (seini) -- (sevfj);
  \draw[conflict] (seini) -- (secfj);
  
  \draw ($ (node1.north east) + (0,0) $) rectangle ($ (secfj.south west) - (6mm,2mm) $);
  
  \node[right=.8cm of es,align=left,draw,rectangle] (box) {$es$ is followed by the event structure\\ from \cref{fig:ex-three}, where $es$ replaces $\bot$};
  \draw[edge from parent] (es) -- (box);
  \endgroup
\end{tikzpicture}
\caption{$\les^{G,B}_2$}
\label{fig:ex-g2}
\end{subfigure}
\caption{Event structures for the language inclusion hardness proof. We use $\ldots$ to indicate omitted events.}
\label{fig:ex-hardness}
\end{figure}

\begin{proof}
\review{

For an undirected graph $G=(V,F)$ and set $B= \{b_1,\ldots,b_m\} \subseteq F$ we construct event structures $\les^{G,B}_1$ and $\les^{G,B}_2$, such that $\lang(\les^{G,B}_1)\subseteq \lang(\les^{G,B}_2)$ iff $G, B$ satisfy DHC. 
$\les^{G,B}_1$ and $\les^{G,B}_2$ are shown in \cref{fig:ex-hardness} and we present the main arguments why this reduction is correct here.
A detailed, formal proof, as well as an example, are given in \cref{sec:langincproof}.

The idea of the proof is to encode subsets $D$ of $B$ via events $ein_i$ and $eout_i$ with $i \in [1,m]$ in both $\les^{G,B}_1$ and $\les^{G,B}_2$.
Events $ein_i$ are labeled with $lb_i$ and represent $b_i \in D$, whereas $eout_i$ are labeled with $\varepsilon$ and represent $b_i \notin D$.
Furthermore, the cardinality of $D$ is encoded in $\les^{G,B}_1$ via $y$-labeled events $eD_i$.
In contrast $y$-labeled events $eD_i$ in $\les^{G,B}_2$ are not used to count $|D|$, but to differentiate whether or 
not a Hamiltonian cycle is required to show DHC (i.e. whether $|D| \leq \frac{|B|}{2}$).
For every $i \in [1,m]$, event $\mathit{efix}_i$ is used to guarantee the existence of maximal configurations in $\les^{G,B}_2$ with $i$ $y$-labeled events.
In case a Hamiltonian cycle is required to show DHC for some set $D$, 
we encode $G_D$ using the same event structure as in the proof of \cref{thm:membership_np_hard}, 
excluding edges from $D$ via conflicts of events $ein_i$.

Our Hamiltonian cycle encoding used in the proof of \cref{thm:membership_np_hard}
operates on directed edges, but DHC is defined for undirected graphs.
Therefore, we replace every edge $f$ in $G$ with two edges in opposing directions,
 denoted by $\vec{f}$ and $\cev{f}$.
Clearly, every Hamiltonian cycle in the directed version corresponds to a Hamiltonian cycle in the undirected graph.
In order to faithfully represent restricted graphs $G_D$, 
we make sure always to exclude all directed edges corresponding to edges in $D$ 
when looking for a Hamiltonian cycle in $G_D$.

Every subset $D$ of $B$ is encoded by some word in $\lang(\les^{G,B}_1)$ via labels $lb_i$ of events $ein_i$.
Since $ein_i$ is in conflict with $eout_i$ maximal configurations can only include one or the other.
That is, words exactly enumerate all subsets $D$ of $B$. Furthermore, similarly as for the proof of \cref{thm:membership_np_hard}, every $w \in \lang(\les^{G,B}_1)$ contains $n$ times the label $x$.

Membership of words of $\lang(\les^{G,B}_1)$ in $\lang(\les^{G,B}_2)$ only depends on whether the encoded subset $F\setminus D$ induces a Hamiltonian cycle:
Words that encode a $D$ such that $|D| > \frac{|B|}{2}$ are always
in $\lang(\les^{G,B}_2)$, because they are trivially accepted by events $ev_i$. 
In contrast, for $D$ such that $|D| \leq \frac{|B|}{2}$ the event $\mathit{efix}_{|D|}$ is in concflict with $ev_1$, 
thereby preventing a trivial acceptance of words in $\lang(\les^{G,B}_1)$.
Therefore, $x$ labels of words in $\lang(\les^{G,B}_2)$ that encode $D$ such that $|D| \leq \frac{|B|}{2}$ 
must be of events caused by $es$.
The events caused by $es$ exactly encode Hamiltionian cycles in $G_D$, similarly to the proof of \cref{thm:membership_np_hard}.

Since the two cases are exhaustive and cover every subset $D$ of $B$, 
we get $\lang(\les^{G,B}_1) \subseteq \lang(\les^{G,B}_2)$ iff $G$ and $B$ satisfy DHC.
The reduction is polynomial as can be easily observed by the event structures in \cref{fig:ex-hardness}.
}
\qed
\end{proof}

\section{Deciding Language Inclusion}
\label{sec:algorithm}

In this section, we introduce a decision algorithm for the FLES language inclusion problem.
Furthermore, we provide a language preserving translation of event structures into non-deterministic finite automata (NFAs), which allows us to compare our algorithm to NFA language inclusion.
We start by introducing necessary concepts for our decision algorithm.

\paragraph*{Configuration as an event structure}
Given an event structure $\les = \langle E, <, \#, h\rangle$ and a configuration $C$, 
we denote its corresponding event structure as 
$\les^C := \langle C, <_{\lceil C\times C}, \emptyset, h_{\lceil C} \rangle $, 
where $X_{\lceil Y}$ denotes the restriction of $X$ to $Y$.
For ease of presentation, when describing our algorithms, 
we abuse notation and do not differentiate between a configuration and its corresponding event structure.
Furthermore, in the following presentation of the algorithms, we assume that the causality relations and 
labeling functions of configurations $C_i$ for $i \in \{1,2\}$ are implicitly given by the event structure interpretation over $\les_i$.

\review{
\paragraph*{$\varepsilon$-free configurations}
For every configuration $C$, there is a configuration with the same language whose only $\varepsilon$-labeled event is $\bot$.
This $\varepsilon$-free configuration can be obtained simply by removing all $\varepsilon$-labeled events besides $\bot$ 
from its corresponding event structure, in particular from $C$ and $<_{\lceil C\times C}$.
The resulting $\varepsilon$-free configuration has the same language as the initial configuration, 
because the causality relation is transitive.
Furthermore, $\varepsilon$-labeled events do not modify the words and 
thus removing them does not influence the language of the configuration.
Therefore, in order to improve readability, 
from hereon we assume without loss of generality that configurations are $\varepsilon$-free.
We keep the $\bot$ event to improve readability, even though for our purpose this event is not required neither.
Note that $\varepsilon$-labeled events are useful during the construction phase of the event structure representing, for example, 
hidden transitions or non-deterministic choices.
}

\paragraph*{Embeddings}

\begin{figure}[tbp]
\centering
\begin{tikzpicture}[grow'=right, sloped,scale=0.5]
	\node[event,label={above:$\varepsilon$}] (bot1) {$\bot$}
	child {
		node[event,label={above:A}] (e1) {$e_1$}
		edge from parent
	}
	child {
		node[event,label={below:B}] (e2) {$e_2$}
		edge from parent
		child {
			node[event,label={below:A}] (e3) {$e_3$}
			edge from parent
		}
	};
  \node [below=5mm of e2,inner sep=0pt]{\parbox{3cm}{\subcaption{Structure 1\label{fig:ex-str1}}}};
	
	\node[event,label={above:$\varepsilon$},right=4cm of bot1] (bot2) {$\bot$}
	child {
		node[event,label={above:A}] (e4) {$e_4$}
		edge from parent
		child {
			node[event,label={above:B}] (e5) {$e_5$}
			edge from parent
		}
	}
	child {
		node[event,label={below:A}] (e6) {$e_6$}
		edge from parent
	};
  \node [below=5mm of e6,inner sep=0pt] (str2) {\parbox{3cm}{\subcaption{Structure 2\label{fig:ex-str2}}}};

\draw[->,dotted,line width=0.3mm] (bot1) -- (bot2);  
\draw[->,dotted,line width=0.3mm] (e1) to[bend left=10] (e4);
\draw[->,dotted,line width=0.3mm] (e2) to[bend right=5] (e5);
\draw[->,dotted,line width=0.3mm] (e3) to[bend left=10] (e6);
	
\begingroup
  \tikzset{level 1/.style={level distance=2cm},level 2/.style={level distance=2cm},level 3/.style={level distance=2cm}}
	\node[event,label={above:$\varepsilon$},right=4cm of bot2] (bot3) {$\bot$}
	child {
		node[event,label={above:A}] (e7) {$e_7$}
		edge from parent
		child {
			node[event,label={above:B}] (e8) {$e_8$}
			edge from parent
			child {
				node[event,label={above:A}] (e9) {$e_9$}
				edge from parent
			}
		}
	};
\endgroup
\node [right=1cm of str2,inner sep=0pt]{\parbox{3cm}{\subcaption{Structure 3\label{fig:ex-str3}}}};

\draw[->,dashdotted,line width=0.3mm] (bot3) -- (bot2);
\draw[->,dashdotted,line width=0.3mm,rounded corners] (e7) to[bend right=20] ($ (e7)!0.5!(e4) $) to[bend left=20] (e4);
\draw[->,dashdotted,line width=0.3mm] (e8) to[bend right=20] (e5);
\draw[->,dashdotted,line width=0.3mm] (e9) to[bend left=10] (e6);
\end{tikzpicture}
\vspace{-5mm}
\caption{Necessary and sufficient embeddings.\\
$\varphi: \bot \mapsto \bot; e_1 \mapsto e_4; e_2 \mapsto e_5; e_3\mapsto e_6$ is a necessary embedding (dotted arrows).\\
$\varphi: \bot \mapsto \bot; e_7 \mapsto e_4; e_8 \mapsto e_5; e_9 \mapsto e_6$ is a sufficient embedding (dash-dotted arrows).
}
\label{fig:ex-necessary}
\end{figure}


An embedding is a structure-preserving one-to-one mapping between events of two configurations from different event structures. 
We consider two different types of embeddings that vary in their strictness in terms of structure preservation.
\review{Since embeddings are defined between configurations, conflicts do not play a role in these considerations.}
In order to use these embeddings for deciding language inclusion between two FLES, 
we assume that in a step prior to searching for embeddings, the maximal configurations of both $\les_1$ and $\les_2$ are computed.
This can, for example, be achieved with the algorithm presented in \cite{Rodriguez2015}.

In the following we consider two configurations $C_1$ and $C_2$ of two $\leslabel$-labeled FLES $\les_1=\langle E_1, <_1, \#_1, h_1\rangle$ respectively
 $\les_2=\langle E_2, <_2, \#_2, h_2\rangle$.

\begin{definition}[Necessary Embedding]
A mapping $\varphi: C_1 \rightarrow C_2$ is a \emph{necessary embedding} if
A) $\varphi$ is bijective, B) $\forall e \in C_1: h_1(e) = h_2(\varphi(e))$, and C) $\forall e \in C_1: \neg \big( e (<_1 \cup <_2^\varphi)^+ e \big)$, 
%
%
where $.^+$ denotes transitive closure and $<_2^\varphi$ denotes the relation $<_2$ mapped to the events of $C_1$.
Formally $<_2^\varphi:=\{(\varphi^{-1}(e_1),\varphi^{-1}(e_2)) \mid \exists e_1,e_2 \in C_2.\ e_1 <_2 e_2 \}$.
For a necessary embedding $\varphi$ from $C_1$ to $C_2$, we write $C_1~\sim^\varphi_N~C_2$.
We write $C_1~\sim_N~C_2$ if there exists a necessary embedding $\varphi$ such that $C_1~\sim^\varphi_N~C_2$.
\end{definition}

A necessary embedding implies that the two configurations have a common word, 
by requiring they have the same number of events with the same labels and 
that their partial orders are not contradicting each other.
\review{Note that the relation $\sim_N$ is symmetric, since for a necessary embedding 
$\varphi: C_1 \rightarrow C_2$, $\varphi^{-1}$ is a necessary embedding from $C_2$ onto $C_1$
}
\begin{example}
Consider the configurations in \cref{fig:ex-str1,fig:ex-str2}. 
There are only two label-preserving bijections between the configurations:
$\varphi_1: \bot \mapsto \bot; e_1 \mapsto e_6; e_2 \mapsto e_5; e_3\mapsto e_4$ and 
$\varphi_2: \bot \mapsto \bot; e_1 \mapsto e_4; e_2 \mapsto e_5; e_3\mapsto e_6$.

The mapping $\varphi_1$ is not a necessary embedding, since $e_2 (<_1 \cup <_2^\varphi) e_2$, which violates C).
To see this, consider the chain of events $e_2<_1 e_3 <_2^\varphi e_2$, where $e_3 = \varphi^{-1}(e_4), e_2 = \varphi^{-1}(e_5),$ and $e_4 <_2 e_5$.
In contrast, $\varphi_2$ is a necessary embedding and a witness to the common word ABA of both configurations.
\end{example}

\begin{restatable}{lemma}{lemmasimempty}
\label{lem:simempty}

Let $C_1$ and $C_2$ be maximal configuration of FLES $\les_1$ respectively $\les_2$.\\
$C_1~\sim_N~C_2$ if and only if $\mathcal{L}(C_1) \cap \mathcal{L}(C_2) \neq \emptyset$.

\end{restatable}

The following corollary gives rise to a termination criterion of the decision algorithm.
If we find a configuration $C$ in $\les_1$, such that there exists no configuration in $\les_2$ that shares a word with $C$, 
we can abort the search and report non-inclusion.

\begin{corollary}

Let $C,C_1,\ldots,C_n$ be configurations \review{such that $C\neq \emptyset$}.
If $(\forall i = 1, \ldots, n : C_i~\nsim_N C)$ then $\mathcal{L}(C) \nsubseteq \bigcup_{i=1}^n \mathcal{L}(C_i)$.

\end{corollary}


The second type of embedding has a stronger requirement on structure preservation.
Intuitively, it requires that the source of such an embedding is at least as strict in terms of causality as the target.

\begin{definition}[Sufficient Embedding]
A mapping $\varphi: C_1 \rightarrow C_2$ is a \emph{sufficient embedding} if
A) $\varphi$ is bijective, B) $\forall e \in C_1:\ h_1(e) = h_2(\varphi(e))$, and C) $\forall e_1,e_2 \in C_1:\ \varphi(e_1) <_2 \varphi(e_2) \implies e_1 <_1 e_2$.
%
%
If there exists a sufficient embedding $\varphi$ from $C_1$ to $C_2$, we write $C_1~\sqsubset^{\varphi}_S~C_2$.
We write $C_1~\sqsubset_S~C_2$ if there exists sufficient embedding $\varphi$, such that $C_1~\sqsubset^{\varphi}_S~C_2$.

\end{definition}

A sufficient embedding is a witness to language inclusion between configurations.
The reason to work with two kinds of embeddings is that we can construct necessary embeddings
using a backtracking algorithm.
It is easy to check whether a necessary embedding is also sufficient, 
whereas it is not straight forward to construct a sufficient embedding from scratch.

\begin{example}
Consider the configurations in \cref{fig:ex-str2,fig:ex-str3}. 
The mapping $\varphi_1: \bot \mapsto \bot; e_7 \mapsto e_4; e_8 \mapsto e_5; e_9 \mapsto e_6$ is a sufficient embedding.
The only non-trivial causality to check is $e_4 <_2 e_5$, for which we have $\varphi_1^{-1}(e_4) = e_7 <_3 e_8 = \varphi_1^{-1}(e_5)$.
In contrast, $\varphi_2: \bot \mapsto \bot; e_7 \mapsto e_6; e_8 \mapsto e_5; e_9 \mapsto e_4$ is not a sufficient embedding,
since in this case $e_4 <_2 e_5$ and $\varphi_1^{-1}(e_4) = e_9 \nless_3 e_8 = \varphi_1^{-1}(e_5)$.
This shows that the language of the event structure in \cref{fig:ex-str3} is included in language of the event structure in \cref{fig:ex-str2}.
\end{example}

The following Lemma provides a connection between sufficient embeddings and language inclusion.
In case there exists a sufficient embedding, the respective languages are included.

\begin{restatable}{lemma}{lemmasufficient}
\label{lem:sufficient}

Let $C_1$ and $C_2$ be maximal configurations of FLES $\les_1$ and $\les_2$ respectively.
If $C_1~\sqsubset_S~C_2$ then $\mathcal{L}(C_1) \subseteq \mathcal{L}(C_2)$.

\end{restatable}

The converse statement is not always true.
To see this, consider a configuration $C_1 = \{\bot,e_1,e_1'\}$ such that $e_1$ and $e_1'$ are concurrent and $h(e_1) = h(e_1') = A$.
Furthermore, consider a configuration $C_2 = \{\bot,e_2,e_2'\}$ such that $e_2$ and $e_2'$ are sequential and $h(e_2) = h(e_2') = A$.
Clearly, the configurations have the same language $\{AA\}$.
However, there is no sufficient embedding from $C_1$ to $C_2$.

Our decision algorithm performs an additional refinement step in such a case and concludes language inclusion only after checking the refined configurations.
In \cref{app:detailed}, we provide a proof that in the case of unique labels, the converse statement also holds.

\paragraph*{Splits}

Our language inclusion decision algorithm continuously performs configuration refinement steps that we call splits.
To be precise, we refine the causality relation of its corresponding event structure.

\begin{definition}[Split]
Let $C$ be a configuration of event structure $\langle E,<,\#,h\rangle$ and let $e_1,e_2\in C$ be two concurrent events.
The \emph{split} of $C$ on $e_1$ before $e_2$ is 
$C_{e_1<e_2}:=\langle C,(<\cup\{(e_1,e_2)\})^+_{\lceil C\times C},\emptyset,h_{\lceil C}\rangle $ where ${.}^+$ denotes transitive closure. 

\end{definition}

\review{
A split on two concurrent events $e_1$ and $e_2$ simply adds an additional ordering constraint between the two events.
In our algorithm, we always split both ways, creating two new configurations that order concurrent events $e_1$ and $e_2$ one way and the other.
Note that in order to avoid duplication of events, in practice splits can be implemented via additional, optional causalities on the event structure.
The following lemma states that splitting a configuration in both ways produces two new configurations with languages whose union is the original language.

\begin{restatable}{lemma}{lemmasplitlang}
\label{lemma:splitlang}

Let $C$ be a configuration and $e_1, e_2 \in C$ be concurrent events, then 
$\lang(C)~=~\lang(C_{e_1<e_2})~\cup~\lang(C_{e_2<e_1})$.
If $h$ is injective (labels are unique), then $\lang(C_{e_1<e_2})~\cap~\lang(C_{e_2<e_1})~=~\emptyset$.

\end{restatable}
}

The following lemma guarantees progress of our algorithm.
It states that if we find a necessary, but not sufficient embedding, 
there are events that can be used to split $C_1$.
The goal is that after a finite number of splits 
a sufficient embedding can be established.

\begin{restatable}{lemma}{lemmasplitwitness}
\label{lemma:splitwitness}

Let $C_1, C_2$ be maximal configuration of FLES $\les_1$ respectively $\les_2$.\\
Furthermore, let $C_1 \sim^{\varphi}_N C_2$ and $C_1 \nsqsubset^{\varphi}_S C_2$.
Then there are concurrent events $e,e' \in C_1$, such that $\varphi(e) <_2 \varphi(e')$.

\end{restatable}

\subsection{Language Inclusion Decision Algorithm}

\begin{algorithm}
\SetAlgoLined\DontPrintSemicolon

\SetKwFunction{FCheck}{Check}
\SetKwProg{Fn}{Function}{:}{}

\SetKwFunction{FSufficient}{SuffOrSplit}
\SetKwProg{Fn}{Function}{:}{}

\KwIn{Finite, labeled Prime Event Structures $\les_1$ and $\les_2$}

\KwResult{$\lang(\les_1) \subseteq \lang(\les_2)$}

$\{C_1^1,\ldots,C_1^n\}, \{C_2^1,\ldots,C_2^m\} \gets $ all maximal configurations of $\les_1$ respectively $\les_2$\;


\Return $\bigwedge_{C_1 \in \{C_1^1,\ldots,C_1^n\}}$ \FCheck{$C_1,\{C_2^1,\ldots,C_2^m\}$}\;

\Fn{\FCheck{$C_1$, $\{C_2^{i_1},\ldots,C_2^{i_l}\}$}}{
\KwResult{$\lang(C_1) \subseteq \bigcup_{j=1}^l \lang(C_2^{i_j})$}
$\var{Candidates} \gets \{C_2^{i_1},\ldots,C_2^{i_l}\}$\;
\ForEach{$C_2 \in \{C_2^{i_1},\ldots,C_2^{i_l}\}$}{
	\If{$\exists\varphi: C_1\rightarrow C_2.\ C_1~\sim^\varphi_{N}~C_2$\label{line:nembed}} {
		\Return \FSufficient{$C_1,C_2,\varphi,Candidates$}\;
	}
	\Else{
		$Candidates \gets Candidates \setminus \{C_2\}$\;
	}
}
\Return $\false$\Comment*[f]{counter-example $C_1$}\; \label{line:no_necessary}
}
\Fn{\FSufficient{$C_1,C_2,\varphi,Candidates$}}{
\If{$C_1 ~\sqsubset^\varphi_{S}~C_2$}{
	\Return $\true$\;
} 
Let $e,e' \in C_1$ be concurrent and $\varphi(e) <_2 \varphi(e')$\Comment*[f]{always exist (\cref{lemma:splitwitness})}\;
\Return \FSufficient{$C_{e < e'},C_2,\varphi,Candidates$} $\wedge$ \FCheck{$C_{e' < e},Candidates$}\; \label{line:sufficient_recursive}
}

\caption{Language inclusion decision algorithm}
\label{algo:comparison}

\end{algorithm}

We present our decision algorithm in \cref{algo:comparison}. 
Inputs to the algorithm are finite, labeled prime event structures $\les_1$ and $\les_2$.

The first step of the algorithm is to \review{calculate the maximal configurations of the event structure}, 
which can be done with the algorithm described in \cite{Rodriguez2015}.
For every maximal configuration $C_1$ of $\les_1$, the function \FCheck{} attempts to show that $\lang(C_1)$ is a subset of $\lang(\les_2)$.
This is achieved by searching for sufficient embeddings from (refined versions of) $C_1$ to maximal configurations of $\les_2$.

In order to construct candidate sufficient embeddings, 
in \cref{line:nembed} the algorithm attempts to construct necessary embeddings, using \cref{algo:necessary}.
In the following line, function \FSufficient{} checks whether a necessary embedding $\varphi$ is also sufficient.
This can be done by checking $\forall e \in C_2: \forall e' \in \dsucc(e): \varphi^{-1}(e) $ is not concurrent with $\varphi^{-1}(e')$.
In case $\varphi$ is not a sufficient embedding, such a pair of events is guaranteed to exist by \cref{lemma:splitwitness}.
For efficiency, this check can already be done during construction of the necessary embedding.

In case $\varphi$ is not sufficient, \cref{lemma:splitwitness} guarantees the existence of a pair of concurrent events that can be split.
The resulting split configurations are recursively checked for language inclusion in \cref{line:sufficient_recursive}.
\cref{lemma:splitwitness} guarantees us that $\varphi$ is a necessary embedding for one of the splits (say $C_{e < e'}$).
Therefore, for $C_{e < e'}$ we do not need to construct a new necessary embedding again, 
but can immediately check whether $\varphi$ is a sufficient embedding for $C_{e < e'}$.

In case no necessary embedding can be found for some configuration $C_1$ and its candidates, according to \cref{lem:simempty}, 
we can conclude $\lang(C_1) \cap \lang(\les) = \emptyset$, i.e. all words in $C_1$ are counter-examples to language inclusion.
Once \cref{line:no_necessary} is reached we know that $C_1$ does not share any word with any $\{C_2^{1},\ldots,C_2^{m}\}$, therefore $C_1$ is a counter-example to language inclusion.

The algorithm terminates, because the notions of necessary and 
sufficient embedding collapse in case the configuration contains only a single trace,
which is the case when the causality relation is a total order on the events of the configuration (see \cref{lemma:totalorder}).

As the algorithm recursively searches for sufficient embeddings, 
for efficiency, we can reduce the set of candidate configurations, 
because in case there is no necessary embedding between two configurations, 
there is clearly also no necessary embedding between any of their split configurations.

\begin{algorithm}[tb]
\SetAlgoLined\DontPrintSemicolon

\SetKwFunction{FContainsCycle}{ContainsCycle}
\SetKwProg{Fn}{}{:}{}

\SetKwFunction{FNEmbedding}{NEmbedding}


\KwIn{Configurations $C_1,C_2$}
\KwResult{$\varphi:C_1\rightarrow C_2$ if $C_1\sim^\varphi_N C_2$, $\none$ otherwise}

\If{$\exists x\in\leslabel.\ |\{e\in C_1 \mid h_1(e)=x\}| \neq |\{e\in C_2 \mid h_2(e)=x\}|$\label{line:trivialfalse}}{
\Return $\false$\;
}
\Return \FNEmbedding{$\{\bot_1\}, \{\bot_1 \mapsto \bot_2\}$}\;

\Fn{\FNEmbedding{$\frontier, \varphi$}}{
\KwIn{$\frontier$ stack of events $C_1$}
\KwIn{$\varphi:C_1 \rightarrow C_2$ (partial mapping)}
\KwResult{$\varphi:E_1\rightarrow E_2$ if $C_1\sim^\varphi_N C_2$, $\none$ otherwise}

\If{$\frontier=\emptyset$\label{line:frontierempty}}{
  \Return $\varphi$\;
}
$e \gets \frontier.\pop()$\;
\ForEach(\Comment*[f]{direct successors of $e$ in $C_1$}){$e'\in \dsucc_{\les^{C_1}}(e)$\label{line:addfrontier}}{
	$\frontier.\push(e')$\;
}
\ForEach{$e''\in C_2$ such that $h_1(e)=h_2(e'')\wedge e''\notin \range(\varphi)$}{
  $\varphi' \gets \varphi\cup\{e\mapsto e''\}$\;
  \If(\Comment*[f]{check for cycle}){$\nexists e_1,e_2\in E_1.\ e_1 <_1 e_2 \wedge \varphi'(e_2) <_2 \varphi'(e_1)$\label{line:cycle}}{
    $\varphi''\gets\FNEmbedding{$\frontier, \varphi'$}$\;\label{line:recursion}
    \If{$\varphi''\neq \none$}{
      \Return $\varphi''$\;
    }
  }
}
\Return $\none$\;
}

\caption{$\sim_N$ decision algorithm}
\label{algo:necessary}

\end{algorithm}

We present the algorithm to construct necessary embeddings in \cref{algo:necessary}.
Intuitively, the algorithm is a combined depth first search over the causality relation, 
as well as the space of possible bijective, label-preserving mappings.

The algorithm starts by dismissing configurations that can never have a necessary embedding because the number of events with the same label differs (\cref{line:trivialfalse}).
The actual embedding is established with the recursive function \FNEmbedding{}. The recursion maintains a frontier of events that are yet to be explored and a partial mapping of already explored events. It ends if the frontier becomes empty (\cref{line:frontierempty}). The exploration is done on $C_1$ by adding the successors of the current event $e$ to the frontier (\cref{line:addfrontier}). Then for every event $e''$ in $C_2$ with the same label as $e$ a mapping $\varphi'$ is created and a cycle check performed. If this mapping does not introduce a cycle we recurs on it (\cref{line:recursion}). The first valid (not $\none$) mapping that is returned by a recursion is returned. If no such mapping is found, then $\none$ is returned.
The cycle check in \cref{line:cycle} basically checks if the two causality relations $<_1$ and $<_2$ are compatible for the mapped events, 
in the sense that the events can be brought in an order that respects both causality relations. 
The procedure can be implemented using any of the well known cycle detection algorithms over the graph with nodes being events of $C_1$ and edges being 
causalities $<_1 \cup <^{\varphi}_2$.

The worst-case runtime of the decision algorithm is exponential in $O(2^{2n})$, 
where $n = |E_1| + |E_2|$, which is not surprising for an algorithm solving a $\Pi^p_2$ hard problem.
There are two dominant factors of the exponential complexity.

First, the number of maximal configurations can be exponential in $n$ and 
\cref{algo:comparison} potentially has to compare all pairs of maximal configurations.
The CCNFS benchmark in \cref{sec:experiments} is an example for an event structure with an exponential number of maximal configurations in $n$.

Second, the number of mappings between configurations that need to be considered as candidates for necessary embeddings can be exponential in $n$.
That is, algorithm \cref{algo:necessary} has worst case runtime exponential in $n$.
Note that for a fixed mapping, the algorithm performs a linear search over the configuration and the combined causality relation.

Note that the number of possible embeddings decreases with the number of calls to \mbox{\FCheck} and 
the size of maximal configurations decreases relative to $n$ with the number of maximal configurations.
Therefore, the amortized runtime should be much better than the worst case complexity.

\subsection{Automaton Based Language Inclusion}
\label{sec:automaton}
We provide a language preserving encoding of event structures into non-deterministic finite automata (NFA).
The encoding allows us to compare our algorithm to well researched language inclusion algorithms in our evaluation (\cref{sec:experiments}).

The encoding has a state for every configuration of the event structure.
There is a transition between two states, if the difference between the corresponding configurations is just one event.
The transition is labeled with the label of that event.
In essence, the encoding is an automaton representation of what is known as the configuration structure of a prime event structure \cite{Glabbeek2009}.

\begin{definition}[Automaton Encoding]
Let $\les = \langle E, <, \#, h\rangle$ be a finite prime event structure with labels $\leslabel$.
We define the non-deterministic finite automaton $\aut^{\les} =  \langle Q^{\les},\Omega^{\les},\delta^{\les},q^{\les}_0,F^{\les}\rangle$ as 
$Q^{\les} = \{q^{\les}_C \mid C \text{ is a configuration of } \les\}$, $\Omega^{\les} = \leslabel$, $(q^{\les}_{C_1},\sigma,q^{\les}_{C_2}) \in \delta^{\les}$ iff there is $e \in E$, such that $C_1 \cup \{e\} = C_2$ and $h(e) = \sigma$, $q^{\les}_0 := q^{\les}_{\{\bot\}}$, and $F^{\les} = \{q^{\les}_C \mid C \text{ is maximal}\}$.

\end{definition}

\vspace{-.3cm}
\begin{restatable}{lemma}{lemmapreserving}
\label{lemma:preserving}

Let $\les$ be a labeled, finite prime event structure, then
$\lang(\les) = \lang(\mathcal{A}^{\les})$.

\end{restatable}

%

The provided encoding is not optimal in general due to conflicts and 
the fact that events of prime event structures are caused in a unique way, 
which is a well known caveat of prime event structures \cite{Winskel1988}.
However, for the family of event structures that consists of the $\bot$ event and $n$ concurrent events (c.f. the proof of \cref{lem:words} in \cref{app:detailed}), 
our encoding contains exactly $2^n + 1$ states, which is one state more than the provably optimal NFA accepting the language of the event structure.
Furthermore, in our experiments, we apply optimized NFA reduction techniques~\cite{mayr2013advanced} before checking language inclusion on automata.

\begin{restatable}{theorem}{exponentialsuccinctness}
\label{lem:words}
There is a family of event structures $\les_n$ with events $E_n$, such that $|E_n| = n + 1$, $|\lang(\les_n)| = n!$, and $|Q^{\les}| = 2^n + 1$.
Every NFA $\mathcal{A}$ with $\lang(\mathcal{A}) = \lang(\les_n)$ has at least $2^n$ states.
\end{restatable}

\section{Application and Evaluation}
\label{sec:experiments}


Our motivation to investigate event structures and language inclusion was model-based mutation testing. 
The goal of \emph{model-based testing} (MBT) is to derive test-cases from a model of a system. 
The model may, for example, be a UML state machine and the test may be a sequence of inputs and outputs of the system. The simplest way of obtaining such test cases would be to randomly explore the state machine and record the produced input/output (IO) sequences. 
These tests can then be run against an implementation at a later point.

\emph{Model-based mutation testing} (MBMT) compares the original model to a mutated version of it, where 
a mutation is a small change in the model, such as removing or adding a transition. 
A test case is only generated if an \review{observable} difference between the original and the mutated model can be witnessed. 
This form of test case generation can be easily expressed using language inclusion between the two versions of the system: 
The test is exactly the word that is a member of the mutant, but not of the original.

The application of finite prime event structures to this problem is motivated by three factors.
Firstly, models often use concurrent state machine that synchronize rarely. 
Secondly, mutation analysis on reactive models can be performed by exploring models in bounded segments \cite{Fellner:2019:MMT:3305158.3289256}, 
where a bounded segment refers to all events occurring between two consecutive inputs.
These bounded segments can be represented as finite event structures.
Thirdly, it is desirable express independence in test cases in order to produce minimal test suites 
that do not need to list all variations of a test that differ only in terms of independent events.
Such test cases can be obtained as counter-examples to language inclusion, as discussed in \cref{sec:algorithm}.

To this end, we implemented the presented prime event structure language inclusion algorithm 
in the model-based mutation testing tool MoMuT~\cite{Fellner:2019:MMT:3305158.3289256}.
MoMuT accepts models written as object-oriented action systems (OOAS).
The models can be understood as labeled transition systems, where labels are either observable, controllable or hidden ($\varepsilon$).
OOAS models can model highly concurrent systems.
In order to construct test cases for such concurrent models efficiently, 
we need to apply partial order reduction during model exploration.
In \cite{Rodriguez2015} a partial order reduction based algorithm for 
constructing labeled prime event structures from transition systems is given.
We implemented this algorithm, using a static dependency relation based on variable reads and writes, 
and use it during model exploration, obtaining finite labeled event structures representing bounded segments of the model.
Each segment corresponds to all output or hidden transitions following some input until either a new input is required 
by the model to progress further or the model is in a terminating state.
We operate on models that exclude infinite sequences of outputs or hidden transitions.
Thus, the discussed segments are indeed bounded in our case.

The event structures constructed during partial order reduction are labeled with (potentially hidden) transitions of the explored model.
However, for mutation analysis, we want to find observable differences between event structures for given controllable inputs, 
in contrast to any difference in transition labels.
Therefore, in addition to using transition labels during partial order reduction, 
we use projected, visible (input \& output) labels and 
perform language inclusion on the languages over the latter kind.

During model exploration, which is described in detail in \cite{Fellner:2019:MMT:3305158.3289256}, 
we construct event structures $\les_B$ of the original model and $\les_A$ of mutants, representing 
bounded segments (as described above) following the same sequence of inputs.
We perform language inclusion checks $\lang(\les_A) \subseteq \lang(\les_B)$ using \cref{algo:comparison} 
to decide whether the corresponding mutant is killed by the sequence of inputs 
and a test case can be produced.

In our experimental evaluation, we report measurements of these language inclusion checks 
during test case generation on a sequence of benchmark models.
For comparison and in addition to event structure based language inclusion, 
we perform language inclusion via automaton encoding, as described in \cref{sec:automaton}.
To this end, we encode the produced event structures as NFAs and check language inclusion using the tool RABIT~\cite{mayr2013advanced}. 

\subsection{Benchmarks and Results}

We use the following benchmarks for our experimental evaluation.
All benchmark models, scripts to instantiate the models for any parameter value, 
and the version of MoMuT used in the experiments can be found in the publicly available artifact of this paper \cite{andreas_fellner_2019_3514619}, 
which can be run with the virtual machine provided in \cite{dietsch_daniel_2019_3533104}.

\begin{itemize}
	\item The \textbf{Paxos$(n,m,k)$} benchmark models the Paxos distributed consensus protocol \cite{Lamport2001} with $n$ proposers, $m$ acceptors, and $k$ learners.
	The protocol specifies how the different actors can exchange certain messages to achieve consensus on some proposed value.
	The actions of the actors are largely independent of each other, which introduces lots of concurrency to the model.
	Furthermore, test cases extracted from our method should be interesting to concertize and run against implementations of the Paxos algorithm.
	\item The \textbf{Semaphore$(n)$} benchmark models $n$ threads that are synchronized by a semaphore. 
	Exactly $n-1$ threads are allowed to enter and compute in a critical section at the same time. 
	The amount of parallelism of this model is proportional to $n$.
	Furthermore, the model exhibits lots of conflicts, 
	as all operations on the semaphore are in conflict with each other.
	\item The \textbf{ParSum$(n)$} benchmark models a parallel summation algorithm.
	The sequence $0,\ldots,n^2{-}1$ is split into $n$ equally sized chunks, which are summed up concurrently. Then the partial results are summed up centrally when all parallel threads are finished.
	\item The \textbf{CCNFS$(n)$} benchmark models a system with $n$ events and unique labels, such that the $2i'th$ event is in conflict exactly with the $2i{+}1 'th$ event.
	Every set of independent events induces an event resetting the state.
	This benchmark is interesting, because its number of maximal configurations $2^{\lfloor\frac{n}{2}\rfloor}$ 
	(each configuration contains either $2i$ or $2i{+1}$ for each $i = 1 \ldots \lfloor\frac{n}{2}\rfloor$)
	is exponential in the number of events $n$.
	Due to the high number of maximal configurations, this benchmark is challenging for our algorithm.
	\item The \textbf{AllPar$(n)$} benchmark models a system with $n$ independent events and unique labels.
	The benchmark is the ideal case for our algorithm, because its event structure consists of only one maximal configuration with all events in parallel.
	In contrast, the benchmark is a very bad case for NFA language inclusion, as the smallest NFA to encode all permutations of $n$ symbols is exponential in $n$ (\cref{lem:words}).
	\item The \textbf{Sharing$(n,m)$} benchmark models a system that has $n$ different prefixes that all share the same suffix of length $m$.
	The benchmark particularly exhibits the well known shortcoming of event structures not being able to encode shared causes.
	The NFA is able to express the common suffix more succinct in comparison to the event structure.
\end{itemize}

\newcommand{\ms}{}
\newcommand{\s}{$\cdot 10^3$}
\newcommand{\TO}{TO}
\newcommand{\NA}{-}
\newcommand{\TOO}{\TO & \TO}
\newcommand{\TOOOO}{\TOO & \TOO}
\newcommand{\TOOOOO}{\TO & \TOOOO}

\newcommand{\inclusionheader}{\textbf{Inclusion}}
\newcommand{\noninclusionheader}{\textbf{Non-Incl.}}
\newcommand{\timeheader}{\mbox{Time}}
\newcommand{\countheader}{\mbox{Num}}
\begin{table}[tb]
	\captionsetup{justification=centering}
	\centering
	\begin{tabular}{l | r r r r r r | r r r r r}
	\textbf{Name} & $\mathbf{|\les_B|}$ & \textbf{PC} & \multicolumn{2}{r}{\inclusionheader} & \multicolumn{2}{r}{\noninclusionheader} & $\mathbf{|\aut^{\les_B}|}$ & \multicolumn{2}{r}{\inclusionheader} & \multicolumn{2}{r}{\noninclusionheader}  \\
	 & & & \timeheader & \countheader & \timeheader & \countheader & &  \timeheader & \countheader & \timeheader & \countheader \\
  \midrule
  Paxos(2,3,1)& 80 & 4.1 & \NA & 0& 4.1\s & 65&  23469& \TOOOO \\
	Paxos(3,6,1)& 716 & 4.1 & \NA & 0& 4.0\s & 94&  \TOOOOO \\
	Semaphore(3)& 9 & 1.2 & \NA & 0& 23.8\ms & 78&  5& \NA & 0& 324.0\ms & 78\\
	Semaphore(11)& 25 & 2.3 & 3.5\ms & 2& 48.9\ms & 84&  9& \NA & 0& 564.2\ms & 80\\
	ParSum(3)& 18 & 2.2 & 1.1\ms & 80& 170.5\ms & 92&  218& 2.5\s & 80& 9.3\s & 84\\
	ParSum(5)& 38 & 3.8 & 342.0\ms & 84& 67.6\ms & 94& \TOOOOO \\
  ParSum(10)& 123 & 8.2 & \NA & 0& 12.6\ms & 23& \TOOOOO \\
  CCNFS(3)& 15 & 1.7 & 3.9\ms & 88& 1.2\ms & 88&  9& 262.7\ms & 88& 610.9\ms & 88\\
	CCNFS(6)& 77 & 2.7 & 501.8\ms & 108& 84.0\ms & 92&   65& 379.7\ms & 108& 1.2\s & 92\\
  CCNFS(10)& 1045 & 4.0 & 303\s & 98& 17\s & 102& 59050 & \TOOOO \\
  AllPar(10) & 12 & 4.0 & 1.4\ms & 56& 8.6\ms & 144&  1025& 10\s & 56& 82\s & 144\\
  AllPar(50) &  52 & 17.3 & 3.2\ms & 44& 93.3\ms & 156& \TOOOOO \\
  AllPar(500) & 502 & 167.3 & 361.9\ms & 40& 7.7\s & 160&  \TOOOOO \\
  Sharing(5,20)& 111 & 1.0 & 7.2\ms & 26& 4.2\ms & 174&  23& 338.1\ms & 26& 1.1\s & 174\\
  Sharing(50,50)&2601 & 1.0 & 1.9\s & 38& 527.2\ms & 162&  53& 255.4\ms & 38& 1.5\s & 162\\
	\end{tabular}
	\centering
  \caption{Benchmark results for language inclusion checks \\
	         $\lang(\les_A) \subseteq \lang(\les_B)$ respectively $\lang(\aut^{\les_A}) \subseteq \lang(\aut^{\les_B})$.}
	\label{tab:results}

\end{table}

We present the results of our experimental evaluation in \cref{tab:results}. 
For every benchmark, we report measurements of language inclusion checks for the largest bounded segment encountered during model exploration.
As described above, every such bounded segment corresponds to all output and hidden transitions following some input transition.
We report measurements of event structure based language inclusion $\lang(\les_A) \subseteq \lang(\les_B)$ and 
automaton based language inclusion $\lang(\aut^{\les_A}) \subseteq \lang(\aut^{\les_B})$.
We separate the results into the cases where language inclusion holds (\textbf{Inclusion}) respectively does not hold (\textbf{Non-Incl.}).
Column $\mathbf{|\les_B|}$ shows the size of $\les_B$ in terms of the number of its events.
Column $\mathbf{|\aut^{\les_B}|}$ shows the size of $\aut^{\les_B}$ in terms of the number of its states.
Columns \countheader~show the number of inclusion checks performed 
on the respective bounded segment (which is the number of mutants relevant in the segment). 
Columns \timeheader~show the average time for the inclusion checks in milliseconds.
Finally, column \textbf{PC} shows a measurement of the degree of concurrency.
For a single configuration $C$, the measurement is defined as $|C|/max_{e\in C}depth(e)$ and 
we report the average measurement of all maximal configurations in the respective event structure.



The reported time for language inclusion of event structures is the time for calculation of the maximal configurations plus the time for the actual language inclusion check. 
The reported time for language inclusion of automata is the time for the language inclusion check on a pre-reduced automaton. 
The construction and minimization of the NFAs is not included.

The results show that our language inclusion algorithm performs well on models with a lot of concurrency, i.e. those with high ParCoeff.
Furthermore, the automaton translation clearly fails in cases with lots of concurrency that are easy for our method (c.f. the \textbf{AllPar} benchmark).
For these examples our algorithm is very useful.
This result is not surprising, since our method exploits concurrency, whereas the NFA encoding does not include any notion of concurrency.
\review{Nevertheless, the result demonstrates that the benefits of exploiting concurrency with our method outweigh  
optimizations and fine-tuning of a well established language inclusion algorithm that has no notion of concurrency.}

However, as the \textbf{Sharing} benchmark shows, the inability of prime event structures to encode shared causes of events is a limitation of the approach.
In contrast, the reduced automaton representation can be significantly more compact than the event structure representation, 
rendering the automaton-based language inclusion superior.




\section{Related Work}

Prime event structures are a widely used formalism to express concurrency of discrete systems \cite{Winskel1988}
that can be obtained from transition systems via the method presented in 
\cite{Rodriguez2015}, or its extended version in \cite{Nguyen2018}.
There are multiple other variants of event structures, 
such as stable event structures \cite{Winskel1988} and flow event structures \cite{Boudol1990}.
Studying language inclusion for these event structure variants is interesting future work.

\review{
Event structure containment based on causality and conflict refinement is considered in \cite{Winskel1988,Winskel2016}.
However, as we demonstrate in our work, causality preservation is not necessary for language inclusion.
In \cite{Glabbeek1989,VanGlabbeek2001a} equivalence of event structures under action refinement is investigated.
This line of research is orthogonal to our approach, as it considers refinement of event structures,
while we compare event structures that can be obtained in multiple different ways.
Moreover, there is almost never language inclusion between an event structure and the event structure with refined actions by design.

Model checking over particular types of event structures has been studied in \cite{Penczek1997} for event structures labeled with atomic propositions and 
in \cite{Madhusudan2003} for event structures labeled with trace languages.
However, the proposed model-checking methods are not based on language inclusion,
which is one of the interesting future directions for our research.
Instead, formulas are directly interpreted over the event structure.

Several formalisms to express concurrency of discrete systems have been proposed and their relationships have been worked out in \cite{Winskel:1995:MC:218623.218630}.
In particular, trace languages and Petri nets are formalisms closely related to event structures, for which languages and language related problems have been studied.
}

The theory of trace languages \cite{Mazurkiewicz1986,Mazurkiewicz1995} studies closure of string languages under independence relations.
\cite{Cerny2017} presents an efficient method to show trace language inclusion over languages defined by 
non-deterministic finite automata.
\review{
In \cite{diekert1997partial} decidability results of rational trace languages are studied.
In particular, it is shown that language inclusion of rational (closed under union, concatenation, and Kleene-star) trace languages is decidable 
if and only if the common independence relation is transitive.
Language membership for context free and regular trace languages was shown to be NP-complete in \cite{bertoni1989membership}.
In \cite{Bouajjani2017}, comparison of concurrent programs via trace languages is studied.
The suggested trace languages abstract the program executions by considering statement ordering, 
as well as read and write accesses on a subset of relevant variables and synchronization primitives.
Trace language refinement is then reduced to assertion checking.
Interestingly, for Boolean programs this refinement check has complexity $\Delta^P_2$ 
for bounded abstraction precision and $\Sigma^P_2$ for unbounded abstraction precision.
In contrast to arbitrary event labels considered in our work, 
the authors of \cite{Bouajjani2017} consider refinement on languages of a more concrete program and dependency model.
}

Our problem is orthogonal to trace language inclusion in three aspects.
Firstly, we do not assume the independence relations of the compared systems to be equal.
Secondly, we do not require the independence relation to be defined over labels.
That is, we can study systems where two labels occur concurrently in one place, 
while the labels occur sequentially in another. This can occur because different events can have the same label.
Finally, in contrast to automata, which are often used to define trace languages, event structures are acyclic.
Therefore, event structures are less expressive than automata.
However, the price of the additional expressivity is that 
trace language inclusion over automata is undecidable in general \cite{Bertoni1982},
whereas our problem is decidable.

\review{
Petri nets are a formalism for concurrent systems that is closely related to event structures \cite{Nielsen1981}.
A manifold of complexity questions have been studied for Petri nets, see \cite{Jones1977,Esparza1996,DBLP:journals/eatcs/EsparzaN94} for surveys.
In particular, language related problems of labeled Petri nets have been studied, 
see \cite{peterson2019petri,Gaubert1999} for an overview over the types of considered languages and complexity results.
Since Petri nets typically describe languages on infinite words and many Petri net related problems are undecidable,
complexity results on language related Petri nets problems focus on establishing the boundary between decidability and undecidability.
Language inclusion and equivalence were shown to be undecidable for a wide range of types of Petri nets \cite{Grabowski1979,Jancar2001,DBLP:journals/eatcs/EsparzaN94}.
Language inclusion is decidable for languages of firing of regular Petri nets \cite{Valk1981} and 
certain types of deterministic Petri nets \cite{Gaubert1999}.
In contrast, language membership is decidable for a large class of Petri nets and language types \cite{peterson2019petri,Hack1976}.
Similarly to trace languages, the additional expressivity of Petri nets 
over finite prime event structures manifests in increased complexity of solving language inclusion.
However, as we demonstrated with our application, finite prime event structures are sufficient for interesting practical problems.
}

Finite asynchronous automata \cite{Fates2013} express concurrent systems succinctly 
in the same spirit as prime event structures.
Furthermore, asynchronous automata accept trace languages \cite{Zielonka1987}.
However, to the best of our knowledge, there is neither an algorithm, 
nor a tool to check language inclusion for (loop free) finite asynchronous automata.

Language inclusion of regular languages is a classic problem of computer science \cite{automata_book,meyer1972equivalence,Friedman1976,stearns1985equivalence}.
Algorithms for the problem are well studied and highly optimized \cite{mayr2013advanced,abdulla2010simulation,Cerny2017}.
However, as we demonstrate in the evaluation section, 
our procedure can outperform these algorithms in the realm of highly concurrent systems.
Adapting methods for classic automaton based language inclusions to event structures is interesting future work.

\section{Conclusion and Future Work}
\label{sec:future}

In this paper we showed that the language inclusion problem between two event structures is computationally hard, 
but our application and evaluation show that there are numerous benchmarks where the use of event structures and their comparison is beneficial.
However, the experiments also manifested a well known shortcoming of prime event structures, namely their inability to succinctly encode shared causes of events.

Interesting future work includes adapting our language inclusion method to different variants of event structures that do not suffer this problem.
Furthermore, we want to study whether our language inclusion procedure can be used to perform model checking over event structures.
Finally, we want to further study the test cases generated for the Paxos distributed consensus algorithm.
Concertizing the resulting test cases and running them against an implementation of the protocol might yield interesting results.

%
%
%
%

\bibliography{references}
\bibliographystyle{splncs04}

\appendix
\crefalias{section}{appsec}

\section{Lemmas and Proofs}
\label{app:detailed}

\begin{lemma}
\label{lemma:totalorder}
Let $C$ be a configuration.
$|T(C)| = 1$ if and only if $|C| = 1$ or for all events $e,e' \in C$ either $(e < e')$ or $(e' < e)$.

\end{lemma}
\begin{proof}

The claim is trivial for $|C| = 1$.

$\Rightarrow$\\
Let $T(C) = \{\langle e_1,\ldots,e_n \rangle\}$.
Since $\langle e_2,e_1,\ldots,e_n \rangle \notin T(C)$, we have that $e_1 < e_2$.
Likewise, we can show that $e_i < e_{i+1}$ for all $i \in [1,n-1]$.
Therefore, $e_i < e_j$ for $i < j$.
Since $C = \{e_1,\ldots,e_n\}$ have shown the claim.

$\Leftarrow$\\
Since $<$ is irreflexive, we immediately get that $|T(C)| > 0$.
From the definition of $T(C)$ follows that $e < e'$ then $e$ appears earlier in any trace in $T(C)$ than $e'$.
Since we assume that every pair of events is ordered by $<$, we have that the order in any trace is fully fixed.
In other words, there can only be a single trace. 
\qed

\end{proof}

\begin{corollary}
$|T(C)| > 1$ if and only if there are concurrent events $e,e' \in C$.
\end{corollary}

\lemmasimempty*

\begin{proof}
$\Rightarrow:$

There is a label preserving bijection $\varphi$, such that $<_1 \cup <_2^\varphi$ is a partial order.
Thus, via a depth first search over the events of $C_1$ with the order $<_1 \cup <_2^\varphi$, we can find a sequence of events $\langle e_1,\ldots,e_n\rangle$, such that
$e_j (<_1 \cup <_2^\varphi) e_i$ implies $j < i$.
It follows that $e_j <_1 e_i$ implies $j < i$, i.e. $\langle e_1,\ldots,e_n\rangle \in T(C_1)$.
Likewise it is the case that $e_j (<_2^{\varphi}) e_i$ implies $j < i$ and by the definition of $<^{\varphi}$, we get $\varphi(e_j) <_2 \varphi(e_i)$, i.e. $\varphi(\langle e_1,\ldots,e_n\rangle) \in T(C_2)$.
Since $\varphi$ is label-preserving, $h(\langle e_1,\ldots,e_n\rangle) = h(\varphi(\langle e_1,\ldots,e_n\rangle)) =: w$, which implies $w \in \lang(C_1)$ and $w \in \lang(C_2)$.

$\Leftarrow:$

Let $t_1 = \langle e_1,\ldots,e_n \rangle$ and $t_2 = \langle e'_1,\ldots,e'_n \rangle$ be traces of $C_1$ and $C_2$ respectively, such that $h(t_1) = h(t_2) \in \lang(C_1) \cap \lang(C_2)$.
We show that the mapping $\varphi: e_i \mapsto e'_i$ is a necessary embedding.
Clearly $\varphi$ is bijective and label-preserving.
We need to show $\forall e \in C_1: \neg e (<_1 \cup <_2^\varphi) e$.
Assume the contrary, i.e. $\exists e \in C_1:  e (<_1 \cup <_2^\varphi) e$.
In order to close the cycle and since $<_1$ and $<_2$ are order relations, there must exist an event $e' \in C_1$, such that $e' <_1 e$ and $e <_2^{\varphi} e'$, i.e. $\varphi(e) <_2 \varphi(e')$.
This implies, that for all $t \in T(C_1)$ we have that $e'$ appears earlier than $e$ and all traces $t\in T(C_2)$ are such that $\varphi(e)$ appears earlier than $\varphi(e')$.
This is a contradiction to $t_1 \in T(C_1), t_2 \in T(C_2)$ and $t_2 = \varphi(t_1)$. 
\qed

\end{proof}

%
%
%
%
%

\begin{restatable}{lemma}{lemmamaxtransfer}
\label{lemma:maxtransfer}

An event $e$ is \emph{maximal} in $C$, if there is no event $e' \in C$, such that $e < e'$.
Let $C_1$ and $C_2$ be configurations such that $C_1~\sqsubset^{\varphi}_S~C_2$.
If $e$ is maximal in $C_1$ then $\varphi(e)$ is maximal in $C_2$.

\end{restatable}

\begin{proof}

Assume there exists $e'$, such that $\varphi(e) <_2 \varphi(e')$.
Since $\varphi$ is a sufficient embedding, we have that $e <_1 e'$, which is a contradiction to $e$ being maximal in $C_1$.
Therefore, no such $e'$ exists.
Since $\varphi$ is bijective, $\varphi(e)$ is maximal in $C_2$. 
\qed

\end{proof}

\begin{restatable}{lemma}{lemmasubconf}
\label{lemma:subconf}

Let $e$ be maximal in configuration $C$, then
$C\setminus \{e\}$ is a configuration and
for every $\langle e_1,\ldots,e_n \rangle \in T(C\setminus \{e\})$ it is the case that $\langle e_1,\ldots,e_n,e \rangle \in T(C)$.

\end{restatable}

\begin{proof}

$C\setminus \{e\}$ is conflict free, because $C$ is conflict free.
$C\setminus \{e\}$ is left closed, because $e$ is maximal, therefore there is no event $e'$ such that $e < e'$.
Note that traces in $T(.)$ contain all events of the configuration.
Therefore, from $\langle e_1,\ldots,e_n \rangle \in T(C\setminus \{e\})$ follows $\{e_1,\ldots,e_n,e\} = C$.
Furthermore, for every $e' \in C$, such that $e' < e$, we have that $e' \in \{e_1,\ldots,e_n\}$, which implies $\langle e_1,\ldots,e_n,e \rangle \in T(C)$.
\qed
\end{proof}

\begin{restatable}{lemma}{lemmasubconf2}
\label{lemma:subconf2}

Let $C$ be a configuration.
$\langle e_1,\ldots,e_n \rangle \in T(C)$ implies that $e_n$ is maximal in $C$.

\end{restatable}

\lemmasufficient*

\begin{proof}

We show the claim by induction on $|C_1| = |C_2| = n$.

In the base case we have $C_1 = \{\bot\}$ and $C_2 = \{\bot\}$.
For these configurations the claim is trivially fulfilled.

Let the induction hypothesis be that the claim holds for all configurations of size $n$.

Assume that $C_1$ and $C_2$ are configurations of size $n+1$, such that $C_1~\sqsubset^{\varphi}_S~C_2$.
We show $T(C_1) \subseteq \varphi(T(C_2))$.
Since $\varphi$ is bijective and label-preserving, $T(C_1) \subseteq \varphi(T(C_2))$ is equivalent to $\lang(C_1) \subseteq \lang(C_2)$.

Let $t = \langle e_1,\ldots,e_{n+1} \rangle \in T(C_1)$.
We need to show that $\varphi(t) \in T(C_2)$. 

Since $e_{n+1}$ is the last event of a trace and we consider only traces that include all events of the configuration, clearly $e_{n+1}$ is maximal in $C_1$.
According to \cref{lemma:maxtransfer,lemma:subconf}, we have that $C_1 \setminus \{e_{n+1}\}$ and $C_2 \setminus \{\varphi(e_{n+1})\}$ are configurations of size $n$.
Furthermore, $\varphi$ restricted to $C_1 \setminus \{e_{n+1}\}$ is a witness for $C_1 \setminus \{e_{n+1}\}~\sqsubset_S~C_2 \setminus \{\varphi(e_{n+1})\}$.
From the induction hypothesis we get $ \varphi(\langle e_1,\ldots,e_n \rangle) \in T(C_2 \setminus \{\varphi(e_{n+1})\})$.
Since $\varphi(e_{n+1})$ is maximal in $C_2$, from \cref{lemma:subconf} follows $\varphi(\langle e_1,\ldots,e_n,e_{n+1} \rangle) \in T(C_2)$.
\qed
\end{proof}

\begin{lemma}
\label{lem:sufficient_2}

Let $C_1$ and $C_2$ be maximal configuration of FLES $\les_1$ and $\les_2$ and for every $e_1 \in C_1$ and $e_2 \in C_2$ it is the case that $|\{e\in C_1 \mid h(e_1) = h(e)\}| = 1$ and $|\{e\in C_2 \mid h(e_2) = h(e)\}| = 1$\\
If $\lang(C_1) \subseteq \lang(C_2)$ then $C_1~\sqsubset_S~C_2$.

\end{lemma}

\begin{proof}

We show the claim by induction on $|C_1| = |C_2| = n$.

In the base case we have $C_1 = \{\bot\}$ and $C_2 = \{\bot\}$.
For these configurations the claim is trivially fulfilled.

Let the induction hypothesis be that the claim holds for all configurations of size $n$.

Let $e_{n+1}$ be a maximal event of $C_1$.
Let $M_{e_{n+1}} := \{e \in C_2 \mid e $ is maximal and $h(e_{n+1}) = h(e)\}$.
We start by showing $\exists e \in M_{e_{n+1}}$, such that
$\lang(C_1 \setminus \{e_{n+1}\}) \subseteq \lang(C_2 \setminus \{e\})$.

Assume the contrary.
That is, for every maximal event $e \in M_{e_{n+1}}$ it is the case that $\lang(C_1 \setminus \{e_{n+1}\}) \nsubseteq \lang(C_2 \setminus \{e\})$.
Then $\lang(C_1 \setminus \{e_{n+1}\}) \nsubseteq \bigcup_{e\in M_{e_{n+1}}} \lang(C_2 \setminus \{e\})$.
That is, there is a word $w \in \lang(C_1 \setminus \{e_{n+1}\})$ such that $w \notin \bigcup_{e\in M_{e_{n+1}}} \lang(C_2 \setminus \{e\})$.
Since $e_{n+1}$ is maximal $w \circ h(e_{n+1}) \in \lang(C_1)$ ($\circ$ denotes concatenation).
Let us define language $\lang^{\circ} := \bigcup_{e\in M_{e_{n+1}}} \lang(C_2 \setminus \{e\}) \circ h(e_{n+1})$ ($\circ$ is applied element-wise).
Clearly, $w \circ h(e_{n+1}) \notin \lang^{\circ}$.

We claim that no word in $\lang(C_2) \setminus \lang^{\circ}$ ends with $h(e_{n+1})$.
Assume there is such a word $w'$.
This word corresponds to a trace ending with some event $e'$ with label $h(e_{n+1})$.
Since $e'$ is the last event in a trace, and we only consider maximal configurations, it needs to be maximal.
Therefore, $e' \in M_{e_{n+1}}$, which is a contradiction to $w' \in \lang(C_2) \setminus \lang^{\circ}$.
Since $w \circ h(e_{n+1})$ ends with $h(e_{n+1})$, we have $w \circ h(e_{n+1}) \notin \lang(C_2) \setminus \lang^{\circ}$.

However, clearly $\lang(C_2) = \lang^{\circ} \cupdot (\lang(C_2) \setminus \lang^{\circ})$.
In summary, we found $w \circ h(e_{n+1}) \in \lang(C_1)$ such that 
$w \circ h(e_{n+1}) \notin \lang(C_2)$, which is a contradiction to $\lang(C_1) \subseteq \lang(C_2)$.

Therefore, we can apply the induction hypothesis to $C_1 \setminus \{e_{n+1}\}$ and $C_2 \setminus \{e\}$ for the existing $e \in M_{e_{n+1}}$.
That is, $C_1 \setminus \{e_{n+1}\} ~\sqsubset^{\varphi}_S~ C_2 \setminus \{e\}$.
We extend $\varphi$ to $C_1$ by setting $\varphi(e_{n+1}) := e$.
Clearly, $\varphi$ is bijective and label-preserving.
Furthermore, causality preservation for events other than $e_{n+1}$ carries over from $C_1 \setminus \{e_{n+1}\}$ to $C_1$.
What remains to show is that for all $e\in C$ with $\varphi(e) < \varphi(e_{n+1})$ we have $e < e_{n+1}$.
Since $e_{n+1}$ is maximal, clearly it is not the case that $e_{n+1} < e$.
Assume that $e$ and $e_{n+1}$ are concurrent.
Then there is a word, in which $h(e_{n+1})$ appears earlier than $h(e)$.
Since events are uniquely labeled and $\varphi(e) < \varphi(e_{n+1})$, this word is not in $\lang(\les_2)$,
which is a contradiction to $\lang(\les_1) \subseteq \lang(\les_2)$.
Therefore, $e < e_{n+1}$ and $\varphi$ is a sufficient embedding.
\qed

\end{proof}

\begin{lemma}
\label{lemma:splits}

Let $C$ be a configuration and $e_1, e_2 \in C$ be concurrent events, then
$T(C) = T(\les^C_{e_1<e_2}) \cup T(\les^C_{e_2<e_1})$ and $T(\les^C_{e_1<e_2}) \cap T(\les^C_{e_2<e_1})=\emptyset$.

\end{lemma}

\begin{proof}

From $t \in T(C_{e_1<e_2})$ directly follows $t \in T(C)$, since there is simply one less constraint to fulfill.
Thus, we have $T(C_{e_1<e_2}) \subseteq T(C)$.
Symmetrically we can show $T(C_{e_2<e_1}) \subseteq T(C)$, thus showing $T(C_{e_1<e_2}) \cup T(C_{e_2<e_1}) \subseteq T(C)$.

Let $t = \langle e'_1,\ldots,e'_n \rangle \in T(C)$.
There are $i \neq j$, such that $e'_i = e_1$ and $e'_j = e_2$.
In the case $i < j$, $t$ fulfills all constraints to be a trace of $C_{e_1<e_2}$, thus $t \in T(C_{e_1<e_2})$.
Likewise $j < i$ implies $t \in T(C_{e_2<e_1})$.
In any case $t \in T(C_{e_1<e_2}) \cup T(C_{e_2<e_1})$, showing $T(C) \subseteq T(C_{e_1<e_2}) \cup T(C_{e_2<e_1}) $.

Further, the $T(\les^C_{e_1<e_2})$ and $T(\les^C_{e_2<e_1})$ are disjoint because no trace can fulfill both constraints $e_1<e_2$ and $e_2<e_1$. 
\qed
\end{proof}
\review{
\lemmasplitlang*

\begin{proof}

$\lang(C) = \lang(C_{e_1<e_2}) \cup \lang(C_{e_2<e_1})$ is a direct consequence of \cref{lemma:splits} and the definition of the language of prime event structures.
If $h$ is injective, i.e. labels are unique, then each word in $\lang(C_{e_1<e_2})$ contains the sub sequence $\ldots h(e_1) \ldots h(e_2) \ldots$ whereas each word in $\lang(C_{e_2<e_1})$ contains the different sub sequence $\ldots h(e_2) \ldots h(e_1), \ldots$, which shows that no word can be part of both languages.
\qed
\end{proof}
}
\lemmasplitwitness*

\begin{proof}

Since $\varphi$ is a label-preserving bijection, 
from $C_1 \nsqsubset^{\varphi}_S C_2$ follows that there are events 
$e,e'\in C$, such that $\varphi(e) <_2 \varphi(e')$ and $e \nless_1 e'$.
Towards contradiction assume $e' <_1 e$.
Then we have $e' (<_1 \cup <^{\varphi}_2) e'$, which is a contradiction to $\varphi$ being a necessary embedding.
Therefore, $e,e'$ are concurrent and a witness to the claim.
\qed
\end{proof}

\lemmapreserving*

\begin{proof}

Given a trace $\vec{e} = \langle e_1,\ldots,e_n \rangle$ we define $C^{\vec{e}}_l := \cup_{i=1}^l \{e_i\}$.

Let $w \in \lang(\mathcal{A}^{\les})$, then by the definition of $\mathcal{A}^{\les}$, there is a trace  $\vec{e} = \langle e_1,\ldots,e_n \rangle$ such that $w = \langle h(e_1),\ldots, h(e_n) \rangle$, such that for every $l \leq n: C^{\vec{e}}_l$ is a configuration and $C^{\vec{e}}_n$ is a maximal configuration.
Since all $C^{\vec{e}}_l$ are configurations, we have $\langle e_1,\ldots,e_n \rangle \in T(C^{\vec{e}}_n)$.
Since $C^{\vec{e}}_n$ is maximal, we have $w \in \lang(\les)$.

Conversely, for every $w \in \lang(\les)$, there is a trace  $\vec{e} = \langle e_1,\ldots,e_n \rangle$ such that $w = \langle h(e_1),\ldots, h(e_n) \rangle$ and $\langle e_1,\ldots,e_n \rangle \in T(\les)$.
From $\vec{e} \in T(\les)$ follows that for every $l \leq n: C^{\vec{e}}_l$ is a configuration and that $C^{\vec{e}}_n$ is maximal.
By the definition of $\delta^{\les}$, we have for every $i = 1,\ldots,n-1: (q^{\les}_{C^{\vec{e}}_i},h(e_i),q^{\les}_{C^{\vec{e}}_{i+1}}) \in \delta^{\les}$.
Since $C^{\vec{e}}_n$ is maximal, $q^{\les}_{C^{\vec{e}}_{n}}$ is accepting.
Therefore, $w \in \lang(Q^{\les})$.
\qed
\end{proof}


%
%

\exponentialsuccinctness*

\begin{proof}
The family is given by the $\leslabel$-labeled prime event structures $\les_n := \langle E_n, <_n, \emptyset, h_n \rangle$, where $\leslabel = \{\varepsilon\} \cup \{1,\ldots,n\}$, $E_n = \{\bot\} \cup \{e_1,\ldots,e_n\}$, $<_n := \{(\bot,e_i) \mid i = 1,\ldots,n\}$, and $h_n(e_i) := i$.
$\les_n$ encodes exactly all permutations of $\{1,\ldots,n\}$.
As shown in \cite{Ellul2005}, no NFA with less than $2^n$ states can accept the language of permutations of $n$ symbols,
showing that every $\mathcal{A}$ with $\lang(\mathcal{A}) = \lang(\les_n)$ has at least $2^n$ states.

\noindent Furthermore, the number of all permutations of $n$ symbols is $n!$, showing that $|\lang(\les_n)| = n!$.

\noindent Finally, every non-empty subset of $E_n$ is a configuration.
There are $2^n + 1$ such sets, which are all subsets of $\{e_1,\ldots,e_n\}$ together with $\bot$ plus the set $\{\bot\}$.
Since states of our encoding correspond one-to-one with configurations, we have $|Q^{\les}| = 2^n + 1$.
\qed
\end{proof}

\section{Detailed Proof of the Language Inclusion Reduction}

\membershipnphard*

\review{

\begin{proof}
We proof by reduction of HC to FLES language membership.

We use $s(f)$ and $t(f)$ to denote the respectively source and target vertices of directed edge $f$.

Let $G = (V,F)$ be a directed graph and define $n := |V|$.
We assume that $G$ does not contain any self loops, i.e. edges $f$ such that $s(f) = t(f)$, and that $n>1$.
Apart from a graph with only one vertex, a graph contains a Hamiltonian cycle if and only if the same graph without self-loops contains a Hamiltonian cycle.
The case $n=1$ can be trivially decided and is therefore not considered in our reduction.
A graph with self-loops can be converted into one without self-loops in linear time by removing self-loops from $F$.

Given an integer $j \in [1,n]$, we abbreviate $(j\ \mathsf{mod}\ n) + 1$ by $\jsucc{j}$.
We say that $f$ is connected to $f'$ if $t(f) = s(f')$.
We say that a sequence of edges $\langle f_1, \ldots, f_n \rangle$ is a cycle if for every $j \in [1,n]: f_j$ is connected to $f_{\jsucc{j}}$.
A Hamiltonian cycle of $G$ is a cycle $\langle f_1, \ldots, f_n \rangle$ of $F$, 
such that $\{t(f_j) \mid j \in [1,n]\} = V$.
The decision problem HC: "Does there exist a Hamiltonian cycle of $G$" is NP-hard \cite{karp1972reducibility}.

We provide a polynomially sized (in $|V|+|F|$) $\{\varepsilon,x\}$-labeled event structure $\les^G := \langle E,<,\#,h \rangle$, 
such that $x^n \in \lang(\les^G)$ ($n$ times the letter $x$ concatenated) if and only if there exists a Hamiltonian cycle of $G$. 
The main idea is that configurations of the event structure translate to connected sequences of edges via events that each represent the assignment of an edge to a position in the sequence and events that express connectedness of two successive edges in the sequence.

\noindent The set of events $E$ contains exactly the following \textbf{events}:

\begin{itemize}
 \item The initial event $\bot$
 \item For every edge $f \in F$ and $j \in [1,n]$, there is an event $e_{f,j} \in E$ representing that $f$ is the $j$'th element of a sequence of edges.
 \item For every pair of edges $f,f' \in F$ with $t(f) = s(f')$ and $j\in [1,n]$ 
there is an event $e_{f,f',j} \in E$ representing that $f$ and $f'$ are successive elements of a sequence of edges.
\end{itemize}

\noindent Formally, we define $E$ as follows:

$$E  := \{\bot\} \cup \bigcup_{j = 1}^n \left(\{e_{f,j} \mid f \in F\} \cup \{e_{f,f',j} \mid f,f' \in F \wedge  t(f) = s(f')\} \right)$$

Clearly, $|\les^G| = |E|$ is polynomial in $|F| \leq |G|$.
Note that causality, conflict relation, and labeling are always of at most quadratic size in $|E|$.

\noindent The causality relation $<$ is the smallest partial order relation containing the following \textbf{causalities}:

\begin{itemize}
 \item For every edge $f \in F$ and $j \in [1,n]$: $\bot < e_{f,j}$ representing that every assignment of a single edge to a position in a sequence of edges is allowed.
 \item For every pair of edges $f,f' \in F$ with $t(f) = s(f')$, and every $j \in [1,n]$: $e_{f,j} < e_{f,f',j}$ and $e_{f',\jsucc{j}} < e_{f,f',j}$ representing that successive edges assigned in the represented sequence are connected.
\end{itemize}

\noindent Formally, we define $<$ as follows, where $X^+$ denotes the transitive closure of relation $X$:

$$<  := \left(\bigcup_{j = 1}^n \left(\{(\bot,e_{f,j}) \mid f \in F\} \cup \{(e_{f,j},e_{f,f',j}), (e_{f',\jsucc{j}},e_{f,f',j}) \mid e_{f,f',j} \in E\} \right)\right)^+ $$

\noindent The conflict relation $\#$ is the smallest conflict relation closed under $<$ (i.e. $e \# e' \wedge e' < e'' \Rightarrow e\# e''$) containing the following \textbf{immediate conflicts}:

\begin{enumerate}[label=\textbf{C.\arabic*}]
 \item \label{conf1} For every edge $f \in F$ and $j, k \in [1,n], j \neq k$: $e_{f,j} \# e_{f,k}$ representing that every edge can only be assigned to one position in a sequence.
 \item \label{conf2} For every pair of edges $f, f' \in F, f \neq f'$ and every $j \in [1,n]$: $e_{f,j} \# e_{f',j}$ representing that every position in the sequence can only be assigned once.
 \item \label{conf3} For every pair of edges $f, f' \in F, f \neq f'$, such that $t(f) = t(f')$ and every $j,k \in [1,n]$: $e_{f,j} \# e_{f,k}$ representing that edges assigned to a sequence must have non-overlapping target vertices.
\end{enumerate}

\noindent Formally, $\#$ has the following set of immediate conflicts:

\begin{align*}
\#^i := \bigcup_{j = 1}^n \bigg(\bigcup_{i = 1, i \neq j}^n \Big(\{(e_{f,j},e_{f,i}) \mid f \in F\} &\cup \{(e_{f,j},e_{f',i}) \mid f,f'\in F \wedge f \neq f' \wedge t(f) = t(f')\} \Big) \\
&\cup \{ (e_{f,j},e_{f',j}) \mid f,f' \in F \wedge  f \neq f'\} \bigg)
\end{align*}

\noindent The \textbf{labeling function} $h$ is given as follows:

\begin{itemize}
 \item Every event of the form $e_{f,f',j}$ has label $h(e_{f,f',j}) := x$.
 \item Every other event in $E$ has label $\varepsilon$.
\end{itemize}

\noindent Foramlly, $h$ is defined as follows:
$$h(e) := \begin{cases} x & \mbox{for } e = e_{f,f',j} \in E \\
						  \varepsilon & \mbox{otherwise} \end{cases}$$

We say that a configuration of $\les^G$ represents a sequence of edge $\langle f_1,\ldots,f_m \rangle $ if it includes events $e_{f_1,i_1},\ldots,e_{f_m,i_m}$, where $\{i_1,\ldots,i_m\} \subseteq [1,n]$, for every $j \in [1,m-1]: i_j < i_{j+1}$.

\noindent\textbf{Configurations represent sequences of edges:}
We claim that every configuration (besides $\{\bot\}$) of $\les^G$ represents a sequence of $m \leq n$ edges with pairwise different targets.

The claim follows from the structure of the immediate conflicts:
Due to conflicts \ref{conf1}, a configuration cannot contain events $e_{f,j}$ and $e_{f,k}$ for $j \neq k$.
Therefore, for every edge, a configuration includes one such event or none.
Furthermore, the indices $j$ of events $e_{f,j}$ give rise to a sequence of represented edges.
Due to conflicts \ref{conf2}, every index can only be assigned to one position in the sequence.
There are at most $n$ possible positions.
Therefore, every configuration (besides $\{\bot\}$) represents a sequence of $m \leq n$ edges.
Furthermore, the edges must pairwise different targets due to conflicts \ref{conf3}.

\noindent\textbf{Configurations with events $e_{f,f',j}$ represent sequences of (partially) connected edges:} 

Due to the causes of events $e_{f,f',j}$, we have that a configuration contains $e_{f,f',j}$ 
if and only if it represents a sequence of edges $f_1,\ldots, f_j = f, f_{\jsucc{j}} = f',\ldots,f_m$.

\noindent\textbf{Hamiltonian cycle $\Rightarrow x^n \in \lang(\les^G)$:}

Assume $G$ has a Hamiltonian cycle $\langle f_1,\ldots,f_n \rangle$.
We claim that the set of events $\bot \cup \{e_{f_j,j}, e_{f_j,f_{\jsucc{j}},j} \mid j \in [1,n]\}$ is a maximal configuration.
The set is causally closed, since the edges are connected.
The set is conflict free, because every position is assigned exactly once (no conflicts among \ref{conf1} and \ref{conf2}) and
since the sequence is a Hamiltonian cycle, the targets are not overlapping (no conflict among \ref{conf3}).
The configuration is maximal, since clearly no further event of the form $e_{f,i}$ can be added,
since the assignment of edges to positions is fixed, and no event of the form $e_{f,f',i}$ can be added, 
since these events are directly induced by the assignment of edges to positions in the sequence.
Furthermore, this maximal configuration contains exactly $n$ events labeled by $x$ 
and all other events are labeled by $\varepsilon$, showing $x^n \in \lang(\les^G)$.

\noindent\textbf{$x^n \in \lang(\les^G) \Rightarrow$ Hamiltonian cycle:}

If $x^n \in \lang(\les^G)$, then $\les^G$ has a maximal configuration $C$ that includes $n$ events of the form $e_{f,f',j}$.

Assume that $C$ does not represent a Hamiltonian cycle.
That is, two successive edges in the represented sequence of edges that are not connected, or not all vertices are visited by the sequence.

The former case cannot be true, due to presence of events $e_{f,f',j}$ in $C$, showing that $f$ is connected to $f'$ and the causalities of such events ($e_{f,j}$ and $e_{f',\jsucc{j}}$), implying that $f$ and $f'$ are successive edges in the sequence of edges represented by $C$.

To see why the latter case cannot be true, consider that in order for a connected sequence of $n$ edges not to visit one of the $n$ vertices of the graph, it needs to visit some vertex twice.
That is, it needs to include edges that have the same target vertex.
$C$ can not represent two different edges with the same target due to conflicts \ref{conf3}.
Furthermore, $C$ can not represent the same edge twice, due to conflicts \ref{conf1}.
Therefore, the sequence of edges represented by $C$ can not contain two edges with the same target.

Therefore, the maximal configuration $C$ represents a Hamiltonian cycle in $G$.
\qed
\end{proof}

}
\review{

\label{sec:langincproof}
\lemdyncycle*

\begin{proof}
We prove the claim by reduction of DHC to FLES language inclusion.
Let $G = (V,F)$ be a finite, undirected graph and let $B \subseteq F$.
Let $V = \{v_1,\ldots,v_n\}$, $F = \{f_1,\ldots,f_k\}$, and $B = \{b_1,\ldots,b_m\}$.
Let $\bhalf := \mathsf{floor}(\frac{|B|}{2})$.
Using the same arguments of generality as in the proof of \cref{thm:membership_np_hard}, 
we assume that $G$ does not contain any self loops, i.e. edges $f$ such that $s(f) = t(f)$, and that $n>1$.

For this proof, we will use the same method to encode Hamiltonian cycles of directed graphs into event structures as in the proof of \cref{thm:membership_np_hard}.
Therefore, the first step of our reduction is to encode $G$ as a directed graph 
$\vec{G} = (V,\vec{F})$ with $\vec{F} := \{\vec{f},\cev{f} \mid f = (u,v) \in F \wedge \vec{f} = \langle u,v \rangle \wedge  \cev{f} = \langle v,u \rangle \}$,
where $(u,v)$ and $\langle u,v \rangle$ denote undirected, respectively directed, edges between vertices $u$ and $v$.
A self loop-free graph $G$ has an undirected Hamiltonian cycle if and only if $\vec{G}$ has a directed Hamiltonian cycle. Intuitively this is true, because we introduce for each undirected edge an edge in each direction that connect the vertices in either direction, just as the undirected edge does.

We provide $\{\varepsilon, x,y,lb_1,\ldots, lb_m\}$-labeled event structures $\les^{G,B}_1 = \langle E_1, <_1, \#_1,h_1 \rangle$ and 
$\mathcal{E}^{G,B}_2 := \langle E_2, <_2, \#_2,h_2 \rangle$ such that 
$|\leslabel|$, $|\les^{G,B}_1| := |E_1|$, and $|\les^{G,B}_2| := |E_2|$ are polynomial in $|G| := |V| + |F|$ and
$G$, $B$ has a DHC if and only if 
$\lang(\les^{G,B}_1) \subseteq \lang(\les^{G,B}_2)$. 
$X^+$ denotes the transitive closure of relation $X$.

We define the components of $\les^{G,B}_1$ as follows, where $\#^i_1$ denotes immediate conflicts:

\begin{align*}
	E_1 &:= \{\bot\}\cup \bigcup_{i=1}^n\{ev_i\} \cup \bigcup_{i=1}^m \{ein_i,eD_i,eout_i\}\\
	<_1 &:= \Big( \{(\bot,ev_1)\} \cup  \bigcup_{i=2}^n \{(ev_{i-1},ev_i)\} \cup  \bigcup_{i=1}^m \{(\bot,ein_i),(\bot,eout_i), (ein_i,eD_i)\} \Big)^+\\
	\#^i_1 &:= \bigcup_{i=1}^m \{(ein_i,eout_i) \}\\
  h_1(e) &:= \begin{cases} x & \mbox{for } e \in \{ev_1,\ldots,ev_n\} \\ lb_i & \mbox{for } e = ein_i \\ y & \mbox{for } e \in \{eD_1,\ldots,eD_m\} \\ \varepsilon & \mbox{otherwise} \end{cases}
\end{align*}

For the definition of $\les^{G,B}_2$, we make use of the event structure $\les^G$ 
encoding all Hamiltonian cycles in $\vec{G}$ defined in the proof of \cref{thm:membership_np_hard} and embed it into $\les^{G,B}_2$ using a new root event $es \in E_2$:
\begin{align*}
E^{HC}  := &\bigcup_{j = 1}^n \left(\{e_{f,j} \mid f \in \vec{F}\} \cup \{e_{f,f',j} \mid f,f' \in \vec{F} \wedge  t(f) = s(f')\}         \right)\\
<^{HC}  := &\bigcup_{j = 1}^n \left(\{(es,e_{f,j}) \mid f \in \vec{F}\} \cup \{(e_{f,j},e_{f,f',j}), (e_{f',\jsucc{j}},e_{f,f',j}) \mid e_{f,f',j} \in E^{HC}\} \right) \\
\#^{HC} := & \bigcup_{j = 1}^n \bigg(\bigcup_{i = 1, i \neq j}^n \Big(\{(e_{f,j},e_{f,i}) \mid f \in \vec{F}\} \cup \{(e_{f,j},e_{f',i}) \mid f,f'\in \vec{F} \wedge f \neq f' \wedge t(f) = t(f')\} \Big) \\
        & \phantom{\bigcup_{j = 1}^n \bigg(\bigcup_{i = 1, i \neq j}^n \Big(\{(e_{f,j},e_{f,i}) \mid f \in \vec{F}\}} \cup \{ (e_{f,j},e_{f',j}) \mid f,f' \in \vec{F} \wedge  f \neq f'\} \bigg)  \\
\end{align*}

Finally, we define the components of $\les^{G,B}_2$ as follows, where $\#^i_2$ denotes immediate conflicts:

\begin{align*}
  E_2 :=& \{\bot, es\} \cup \bigcup_{i=1}^n\{ev_i\} \cup \bigcup_{i=1}^m \{ein_i,eout_i,eD_i,\mathit{efix}_{i-1}\} \cup E^{HC}\\
  <_2 :=& \Big( \{(\bot,es),(\bot,ev_1),(\bot,\mathit{efix}_0),(\bot,eD_1)\} \cup  \bigcup_{i=2}^n \{(ev_{i-1},ev_i)\} \cup  \bigcup_{i=1}^m \{(\bot,ein_i),(\bot,eout_i)\} \\
  & \cup \bigcup_{i=1}^{m-1} \{(eD_i,\mathit{efix}_i),(eD_i,eD_{i+1})\} \cup <^{HC} \Big)^+\\
  \#^i_2 :=& \{(es,ev_1)\}\cup \bigcup_{i=1}^m \{(ein_i,eout_i), (eD_i,\mathit{efix}_{i-1}) \} \cup \bigcup_{i=1}^\bhalf \{(\mathit{efix}_i,ev_1) \cup \#^{HC} \} \cup \\
  & \bigcup_{i=1}^m \bigcup_{j=1}^n \{(ein_i,e_{\vec{f_i},j}), (ein_i,e_{\cev{f_i},j}) \}\\
	h_2(e) :=& \begin{cases} x & \mbox{for } e = e_{f,f',j} \in E_2 \mbox{ for some } f,f' \in F \mbox{ and } j \in [1,n] \\ x & \mbox{for } e \in \{ev_1,\ldots, ev_n\}\\ y & \mbox{for } e \in  \{eD_1,\ldots,eD_m\} \\ lb_i & \mbox{for } e = ein_i \\ \varepsilon & \mbox{otherwise} \end{cases}
\end{align*}

From the definition of $\les^{G,B}_1$ and $\les^{G,B}_2$, we can immediately see that the reduction is polynomial.
Notice that the causality and the conflict relation, as well as the number of edges of the graph is always at most quadratic in the number of events and vertices in $G$.



Both event structures $\les^{G,B}_1$ and $\les^{G,B}_2$ encode subsets $D$ of $B$.
In particular, every word $w \in \lang(\les^{G,B}_1) \cup \lang(\les^{G,B}_2)$ encodes a subset $D(w) := \{b_i \mid lb_i \in w\}$ of $B$.

Consider an arbitrary word $w \in \lang(\les^{G,B}_1)$.
We first gather some properties of its symbols and then proceed to show under which condition $w \in \lang(\les^{G,B}_2)$ holds.
Firstly, $w$ must contain the symbol $x$ exactly $|V|$ times due to events $ev_i$.
Secondly, every event $ein_i$ of $\les_1$ causes an event $eD_i$ with label $y$.
Therefore, $w \in \lang(\les^{G,B}_1)$ encodes the cardinality of $D(w)$ by the number of $y$ symbols in $w$.
Finally, since all events $ev_i$ are concurrent to all events $ein_j$ and all events $ein_j$ are pairwise concurrent, 
symbols $lb_i$ and $x$ can appear in any permutation.

To show under which conditions $w \in \lang(\les^{G,B}_2)$ holds, we distinguish whether $y$ occurs in $w$ more often or less or equal than $\bhalf$ times.

Firstly, consider the case that $w$ contains the symbol $y$ more often than $\bhalf$.
That is, $w$ encodes $D(w)$ with $|D(w)| > \bhalf$.
Since DHC does not require the existence of a Hamiltonian cycle in $G_{D(w)}$, the word should be contained in $\lang(\les^{G,B}_2)$ for any graph $G$.
Consider the set of events $C_{D(w)} := \{\bot\} \cup \{ein_j \mid b_j \in D(w)\} \cup \{eout_j \mid b_j \notin D(w)\} \cup \lceil \mathit{efix}_{|D(w)|} \rceil \cup \{ev_1,\ldots,ev_n\}$. 
The set is a maximal configuration, since $\mathit{efix}_{|D(w)|}$ is not in conflict with $ev_1$, due to $|D(w)| > \bhalf$. 
Furthermore, $w \in \lang(C_{D(w)})$ because the $y$-labeled events $eD_i$ are pairwise concurrent with the $x$-labeled events $\{ev_1,\ldots,ev_n\}$ which are in turn pairwise concurrent with the $lb_i$-labeled $\{ein_j \mid b_j \in D(w)\}$ events.
Furthermore, the $lb_i$-labeled events $\{ein_j \mid b_j \in D(w)\}$ are pairwise concurrent to each other.
Therefore, $\lang(C_{D(w)})$ includes exactly all permutations of these symbols, in particular $w \in \lang(C_{D(w)})$.
In summary, for words $w$ which encode $D(w)$, such that $|D(w)| > \bhalf$, we have $w \in \lang(\les^{G,B}_2)$.

Secondly, consider the converse case that $w$ encodes a set $D(w)$ with $|D(w)| \leq \bhalf$.
Any maximal configuration $C_{D(w)}$ such that $w \in \lang(C_{D(w)})$ must include events
$\{\bot\} \cup \{ein_j \mid b_j \in D(w)\} \cup \{eout_j \mid b_j \notin D(w)\} \cup \lceil \mathit{efix}_{|D(w)|} \rceil$.

In contrast to the first case, the event $ev_1$ is in conflict with $\mathit{efix}_{|D(w)|}$, because $|D(w)| \leq \bhalf$.
However, the event $es$ is not in conflict with $\mathit{efix}_{|D(w)|}$.
Hamiltonian cycles are encoded in the sub-event structure following $es$, which can be seen by the arguments presented in the proof of \cref{lem:cycle_red}.
However, due to conflicts of events $ein_i$ with $e_{\vec{f}_i,j}$ and $e_{\cev{f}_i,j}$, 
only events $e_{\vec{f},j}$ and $e_{\cev{f},j}$ can be included in $C_{D(w)}$ with $f \notin D(w)$, i.e. edges of the graph $G_{D(w)}$.
Thus, analogous to the proof of \cref{thm:membership_np_hard}, $C_{D(w)}$ include $|V|$ $x$-labeled events if and only if $G_{D(w)}$ has a Hamiltonian cycle.

In summary, for a word $w \in \lang(\les_1)$, we have $w \in \lang(\les_2)$ if and only if $|D(w)| > \bhalf$ or $|D(w)| \leq \bhalf$ and $G_{D(w)}$ contains a Hamiltonian cycle.
Furthermore, for every $D \subseteq B$, there is a word $w \in \lang(\les_1)$, such that $D = D(w)$.

Therefore, we get $\lang(\les^{G,B}_1) \subseteq \lang(\les^{G,B}_2)$ if and only if $G$ and $B$ satisfy DHC.
\qed
\end{proof}

}

\begin{example}

Consider the graph shown in \cref{fig:dyn_ham_graph}.
The graph and the set of edges $\{b_1,b_2\}$ form a dynamic Hamiltonian cycle.

\Cref{fig:dyn_ham_struct1,fig:dyn_ham_conf} 
show the event structures $\les_1$ and $\les_2$ constructed according to the reduction in the proof of \cref{lemma:dyn_ham_cycle} on this example graph.

The bold events in \cref{fig:dyn_ham_struct1,fig:dyn_ham_conf} show the configurations
that correspond to the Hamiltonian cycle that can be obtained when $b_2$ is removed. 
The configuration where $b_1$ is removed is similar.

The remaining cases are removing none or both of $b_1$ and $b_2$.
In the case that both are removed, the number of $y$ labels is 2 and therefore $eD_2$ has to be part of the configuration, enabling $ev_1,\ldots,ev_4$ to generate 4 times $x$.
In case no $b_i$ is removed, the same configuration as in \cref{fig:dyn_ham_conf} can be used to show existence of a Hamiltonian cycle with small changes to accommodate for the different number of $y$ and the non-existence of $lb_2$.

\end{example}

\begin{figure}[!tbp]
	\begin{subfigure}{0.33\textwidth}
		\centering
		\includegraphics[width=.8\textwidth]{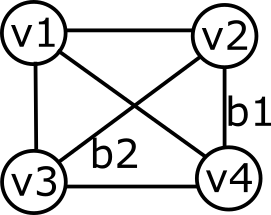}
		\caption{Graph with dynamic Hamiltonian cycle.}
		\label{fig:dyn_ham_graph}
	\end{subfigure}\hfill
	\begin{subfigure}{0.33\textwidth}
		\centering
			\includegraphics[width=.8\textwidth]{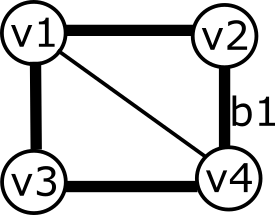}
			\caption{Cycle without $in_2$.}
			\label{fig:cycle1}
	\end{subfigure}\hfill
	\begin{subfigure}{0.33\textwidth}
		\centering
			\includegraphics[width=.8\textwidth]{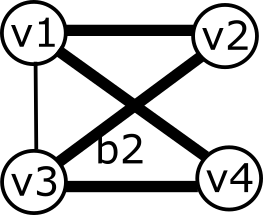}
			\caption{Cycle without $in_1$.}
			\label{fig:cycle2}
	\end{subfigure}
  
	\caption{Graph}
\end{figure}

\centering
\begin{figure}[!tbp]
\begin{tikzpicture}[grow'=right, sloped,scale=0.5,level 1/.style={sibling distance=2cm,level distance=3cm},level 2/.style={sibling distance=2cm,level distance=3cm}]
	\node[event,conf] {$\bot$}
	child {
		node[event,conf,label={above:$x$}] (ev1) {$ev_1$}
		edge from parent
    child {
      node[event,conf,label={above:$x$}] (ev2) {$ev_2$}
      edge from parent
      child {
        node[event,conf,label={above:$x$}] (ev3) {$ev_3$}
        edge from parent
        child {
          node[event,conf,label={above:$x$}] (ev4) {$ev_4$}
          edge from parent
        }
      }
    }
	}
  child {
		node[event,label={above:$lb_1$}] (eb1) {$ein_1$}
		edge from parent
    child {
      node[event,label={above:$y$}] (ebb1) {$eD_1$}
      edge from parent
    }
	}
  child {
		node[event,conf] (ebbb1) {$eout_1$}
		edge from parent
	}
  child {
		node[event,conf,label={above:$lb_m$},yshift=-.3cm] (ebm) {$ein_2$}
		edge from parent
    child {
      node[event,conf,label={above:$y$}] (ebbm) {$eD_2$}
      edge from parent
    }
	}
  child {
		node[event,yshift=-.3cm] (ebbbm) {$eout_2$}
		edge from parent
	};
	
  \draw[conflict] (eb1) -- (ebbb1);
  \draw[conflict] (ebm) -- (ebbbm);
\end{tikzpicture}
\caption{Example Event structure $\les_1$}
\label{fig:dyn_ham_struct1}
\end{figure}\hfill
\centering
\begin{figure}[!tbp]
\centering
\begin{tikzpicture}[grow'=right, sloped,scale=0.5,level 1/.style={sibling distance=2.5cm,level distance=3cm},level 2/.style={sibling distance=2.5cm,level distance=3cm}]
	\node[event,conf] (bot) {$\bot$};
  \node[event,label={above:$lb_1$},right=2.5cm of bot,yshift=3cm] (eb1) {$ein_1$};
  \node[event,conf,label={above:$\varepsilon$},above left=0.5cm of eb1,yshift=-0.3cm] (ebbb1) {$eout_1$};
	\node[event,conf,label={above:$lb_2$},right=2.5cm of bot,yshift=1.5cm] (ebm) {$ein_2$};
  \node[event,above left=0.5cm of ebm,yshift=-0.3cm] (ebbbm) {$eout_2$};
  \draw[edge from parent] (bot) |- (eb1);
  \draw[edge from parent] (bot) -- (ebbb1);
  \draw[edge from parent] (bot) -- (ebm);
  \draw[edge from parent] (bot) -- (ebbbm);
  
  \node[event,conf,right=0.8cm of bot,yshift=-0.2cm] (es) {$es$}
  child {
    node[event,conf,yshift=0.6cm] (efv1v2) {$e_{v1v2,1}$}
    edge from parent[draw=none]
    child {
      node[event,conf] (efv2v4) {$e_{v2v4,2}$}
      edge from parent[draw=none]
      child {
        node[event] (efv3v4) {$e_{v3v4,3}$}
        edge from parent[draw=none]
        child {
          node[event,conf] (efv4v3) {$e_{v4v3,3}$}
          edge from parent[draw=none]
          child {
            node[event,conf] (efv3v1) {$e_{v3v1,4}$}
            edge from parent[draw=none]
            child {
              node[event] (efv2v3) {$e_{v2v3,1}$}
              edge from parent[draw=none]
              child {
                node[event] (efv3v2) {$e_{v3v2,1}$}
                edge from parent[draw=none]
                child {
                  node (dummy) {}
                  edge from parent[draw=none]
                }
              }
            }
          }
        }
      }
    }
  };
  
  \draw[edge from parent] (bot) -- (es);
  \draw[edge from parent] (es) -| (efv1v2);
  \draw[edge from parent] (es) -| (efv2v4);
  \draw[edge from parent] (es) -| (efv3v4);
  \draw[edge from parent] (es) -| (efv4v3);
  \draw[edge from parent] (es) -| (efv3v1);
  \draw[edge from parent] (es) -| (efv2v3);
  \draw[edge from parent] (es) -| (efv3v2);
  \draw[invisible] (efv3v2) --node[auto=false]{\ldots} (dummy);
  
  \draw[conflict] (efv2v3) -- (efv3v2);
  \draw[conflict] (efv3v4) -- (efv4v3);
  
  \draw[conflict] (ebm) -| (efv2v3);
  \draw[conflict] (ebm) -| (efv3v2);
  \draw[conflict] (eb1) -| (efv4v3);
  \draw[conflict] (eb1) -| (efv3v4);
  
  \node[event,conf,below=1cm of efv1v2,xshift=.5cm,label={below:$x$}] (efv1v2efv2v4) {$e_{v1v2,v2v4,1}$};
  \node[event,conf,below=1cm of efv2v4,xshift=1cm,label={below:$x$}] (efv2v4efv4v3) {$e_{v2v4,v4v3,2}$};
  \node[event,conf,below=1cm of efv4v3,xshift=.5cm,label={below:$x$}] (efv4v3efv3v1) {$e_{v4v3,v3v1,3}$};
  \node[event,conf,below=1cm of efv3v1,xshift=1cm,label={below:$x$}] (efv3v1efv1v2) {$e_{v3v1,v1v2,4}$};
  
  \draw[edge from parent] (efv1v2) -- (efv1v2efv2v4);
  \draw[edge from parent] (efv2v4) -- (efv1v2efv2v4);
  \draw[edge from parent] (efv2v4) -- (efv2v4efv4v3);
  \draw[edge from parent] (efv4v3) -- (efv2v4efv4v3);
  \draw[edge from parent] (efv3v1) -- (efv4v3efv3v1);
  \draw[edge from parent] (efv4v3) -- (efv4v3efv3v1);
  \draw[edge from parent] (efv3v1) -- (efv3v1efv1v2);
  \draw[edge from parent] (efv1v2.south east) -- ++(0,-1cm) -| (efv3v1efv1v2);
  
  \node[event,label={80:$x$},below=1.5cm of es] (ev1) {$ev_1$}
  child {
    node[event,label={above:$x$}] (ev2) {$ev_2$}
    edge from parent
    child {
      node[event,label={above:$x$}] (ev3) {$ev_3$}
      edge from parent
      child {
        node[event,label={above:$x$}] (ev4) {$ev_4$}
        edge from parent
      }
    }
  };
  \draw[conflict] (ev1) -- (es);
  \draw[edge from parent] (bot) -- (ev1);
	
  \node[event,conf,label={below:$y$},below=1.5cm of ev1] (ey1) {$eD_1$}
  child {
    node[event,label={below:$y$}] (ey2) {$eD_2$}
    edge from parent
  };
  \draw[edge from parent] (bot) |- (ey1);

  \node[event,above=0.3cm of ey1] (eyy1) {$\mathit{efix}_0$};
  \draw[edge from parent] (bot) -- (eyy1);
  \draw[conflict] (ey1) -- (eyy1);
  \draw[conflict] (eyy1) -- (ev1);
  \node[event,conf,above=0.3cm of ey2] (eyy2) {$\mathit{efix}_1$};
  \draw[edge from parent] (ey1) -- (eyy2);
  \draw[conflict] (ey2) -- (eyy2);
  \draw[conflict] (eyy2) -- (ev1);
  
  \draw[conflict] (eb1) -- (ebbb1);
  
  \draw[conflict] (ebm) -- (ebbbm);

\end{tikzpicture}

\begin{justify}
For space reasons only selected events of the form $e_{f,i}$ and $e_{f,f',i}$ are drawn. For clarity we name the edges $\vec{f}$ and $\cev{f}$ by the vertices they connect, e.g. $v1v2$ and $v2v1$.
\end{justify}
\caption{Example Event structure $\les_2$}
\label{fig:dyn_ham_conf}
\end{figure}

\end{document}